\RequirePackage{fix-cm}
\documentclass[]{svjour3}          

\smartqed  
\usepackage{graphicx}
\usepackage{inputenc}
\usepackage{multirow}
\usepackage{booktabs}
\usepackage{capt-of}
\usepackage{arydshln}
\usepackage{subfigure}
\usepackage{algorithm}
\usepackage{color}
\usepackage{algorithmic}
\usepackage{amsmath}
\usepackage{graphicx}
\usepackage{amsfonts}
\usepackage{rotating}
\usepackage{mhchem}
\usepackage{inputenc}
\usepackage{hyperref}
\usepackage{arydshln,amssymb,color}
\usepackage{lscape}
\usepackage{float}
\usepackage{mleftright} 
\usepackage{verbatim} 

\usepackage{amsmath}
\usepackage[vcentermath]{youngtab}
\usepackage{ytableau}

\usepackage{ bbold }

\DeclareMathOperator*{\argmin}{argmin}
\DeclareMathOperator*{\argmax}{argmax}

\begin{document}
\raggedbottom

\title{Taxonomization of Combinatorial Optimization Problems in Fourier Space}

\titlerunning{Taxonomization of Combinatorial Optimization Problems in Fourier Space}        

\author{Anne Elorza					\and
				Leticia Hernando    \and
        Jose A. Lozano
}


\institute{A. Elorza \at
             Intelligent Systems Group \\
             Department of Computer Science and Artificial Intelligence\\
             University of the Basque Country UPV/EHU\\
             20018 San Sebasti\'an, Spain\\
             \email{anne.elorza@ehu.eus}           
           \and%
			L. Hernando \at
             Intelligent Systems Group \\
             Department of Applied Mathematics and Statistics and Operational Research\\
            University of the Basque Country UPV/EHU\\
             48940 Leioa, Spain\\
              \email{leticia.hernando@ehu.eus}           
           \and
           J. A. Lozano \at
		Basque Center for Applied Mathematics (BCAM) \\
		48009 Bilbao, Spain \\
             Intelligent Systems Group \\
             Department of Computer Science and Artificial Intelligence\\
            University of the Basque Country UPV/EHU\\
             20018 San Sebasti\'an, Spain\\
              \email{ja.lozano@ehu.eus}
              }


\maketitle

%
%
%
%
%
%
%
%

%
%

\begin{abstract}

We propose and develop a novel framework for analyzing permutation-based combinatorial optimization problems, which could eventually be extended to other types of problems. Our approach is based on the decomposition of the objective functions via the generalized Fourier transform. We characterize the Fourier coefficients of three different problems: the Traveling Salesman Problem, the Linear Ordering Problem and the Quadratic Assignment Problem. This implies that these three problems can be viewed in a homogeneous space, such as the Fourier domain. Our final target would be to create a taxonomy of problem instances, so that functions which are treated similarly under the same search algorithms are grouped together. For this purpose, we simplify the representations of the objective functions by considering them as permutations of the elements of the search space, and study the permutations that are associated with different problems.



\keywords{Combinatorial Optimization Problems \and Fourier Transform \and Rankings  \and Permutations}
\end{abstract}

%
%

\section{Introduction}\label{intro}

The field of combinatorial optimization (CO) has mainly developed by providing new algorithms. Even though more theoretical advances are undeniable, many fundamental questions remain still unanswered. The publication of the ``no free lunch theorem'' \cite{wolpert1997no} made clear the necessity of designing and employing problem-specific algorithms, since no algorithm performs better than a random search over the whole set of problems. Consequently, the main target would be to find an association between problems and algorithms, in such a way that, once we know the problem we want to solve, we can use the most efficient algorithm for it. \\   

A first step in this direction would be te create a taxonomy, so that problems showing similar behaviours under the same algorithms are grouped together. The literature already provides several taxonomies based on different criteria. The most recurrent would probably be the one given by complexity theory, which broadly divides the combinatorial optimization problems (COPs) in $P$ and $NP$-hard problems \cite{cook2006p}. A finer classification is given by parametrized complexity, where the complexity of the problem is based on a parameter \cite{downey2012parameterized}. Even so, we could go a step forward and use a  more fine-grain view: we would like to be able to select the best algorithm given a specific instance of a combinatorial optimization problem, as two different instances of the same problem can have different characteristics that make them amenable to different algorithms. \\

When facing such a task, one immediately encounters two obstacles: the immensity of the space of permutation-based functions, as well as the fact that the formulation of the different problems has different meanings. 
For instance, although the linear ordering problem (LOP), the traveling salesman problem (TSP), the quadratic assignment problem (QAP) and the permutation flowshop scheduling problem (PFSP) are defined by means of matrices, the meaning of the elements of these matrices are different according to the problem. In this sense, a problem addresses distances between cities, it copes with flows between factories, or it is about tasks on machines. To overcome this limitation, we propose and develop a novel working framework: viewing the COPs in the Fourier domain. Similarly to elementary landscape analysis \cite{chicano2011methodology}, Fourier analysis produces a decomposition of a given function into a sum of smaller pieces, from which the original function can easily be recovered. Whereas its real-line counterpart is universally-known in the area of applied mathematics \cite{korner1989fourier}, at the present time, this is not the case of the Fourier transform (FT) over permutations \cite{terras1985harmonic,terras1999fourier,terras2012harmonic}. Despite being substantially less noted than its real analogue, the FT over permutations has recently been gaining attention in the computer science field, giving rise to more proposals of applications and algorithms \cite{huang2009fourier,kondor2008group}, such as object tracking or analyzing ranking data. A survey with applications of the group FT, can be found at \cite{rockmore2004recent}. Even though its employment has also reached the area of CO \cite{kondor2010fourier,rockmore2002fast} and its boolean twin, the Walsh transform, has also been considered \cite{christie2016role,goldberg1989genetic}, the use of Fourier analysis has still remained limited. However, the FT seems to offer a suitable space for COPs to be treated in an homogeneous way, via their Fourier decomposition. \\ 

The other difficulty, the arduousness of dealing with a space as big as the permutation-based functions, can be overcome by treating functions as permutations of the elements of the search space. It suffices to observe that most heuristic algorithms do not use the value of the objective functions in the optimization process, but they only need to compare the objective function value of several solutions. Some examples are evolutionary algorithms that use tournament or ranking selection operators, or local search based-algorithms such as tabu search, variable neighborhood search, iterated local search, greedy randomized adaptive search, etc. These algorithms behave the same in two functions $f$ and $g$ such that for all $x$ and $y$ if $f(x) > (<) f(y)$ then $g(x) > (<) g(y)$, i.e. the behaviour of the algorithms only depends on the ranking of the solutions imposed by the function. So real-valued functions can be simplified and treated as rankings (or permutations) of elements of the search space. Thus, we limit the infinite number of possible objective functions to a finite number of permutations. \\


Under this working framework, our first contribution consists in characterizing the Fourier coefficients of functions representing LOP and TSP problems, respectively, and is materialized in two theorems and a number of experiments. We demonstrate very specific properties that the Fourier coefficients of any objective function associated with these problems must satisfy. We also shed light on the opposite path by making experimentation that strongly suggests what kind of coefficients would lead to TSP and LOP problems. Finally, we experimentally analyze general features of the Fourier coefficients of the QAP.  \\


Our second contribution is related to the idea that we adopt of regarding functions as permutations or, in other words, viewing them as rankings of the elements of the search space, where the first item is the solution with the optimal value of the function and the rest are sorted in descending or ascending order (depending on whether the objective is to maximize or minimize, respectively). When one adopts this approach, a range of issues arise. Probably, the most primary one would be: given a specific problem, can we characterize all the rankings that its instances can generate? If not, can we discard any of them? In the case of the Quadratic Assignment Problem, we have answered to this question for low dimensions, and for the Linear Ordering Problem, we have found a restrictive bound on the number of different rankings that it can generate for any dimension. \\

The rest of the paper is organized as follows: Sections \ref{sec::FT_background} and \ref{subsec::COPs_background} present the necessary background related to the Fourier transform and combinatorial optimization problems, respectively. Sections \ref{sec::StructureCOPcoeff} and \ref{sec::rankings} show our contributions, the former regarding the Fourier characterizations of COPs and the latter regarding the study of COPs when considered as rankings. Section \ref{sec::Conclusion} concludes the paper, including proposals for future work. We have also added an appendix, where the most specific and technical aspects are gathered. Appendix \ref{appendix::proofs} contains the proofs of the theorems that are presented in Section \ref{sec::StructureCOPcoeff}, while the other two appendices, Appendix \ref{app:isLOP_isTSP} and Appendix \ref{appendix::impossibleRanking} explain how certain methods and algorithms of Section \ref{sec::rankings} have been developed and implemented.\\

\section{Fourier transform on the symmetric group}\label{sec::FT_background}

The next lines offer a brief overview of the Fourier transform. The interested reader may refer to \cite{sagan2013symmetric} for a deeper insight into the precise algebraic concepts. Since this is a general introduction to the subject, the most technical details, which are needed to prove the theorems of section \ref{sec::StructureCOPcoeff}, can be found in appendix \ref{app:Theorems_Background}.\\

The Fourier transform on the symmetric group comes from a generalization of the well-known transform on the real line to finite groups. Thus, in an initial encounter, it may be useful to understand it in contrast with the real case. Given a function $f:\Sigma_n\longrightarrow º\mathbb{R}$, the Fourier transform decomposes $f$ into certain coefficients by means of a set of base functions (on the real line, sines and cosines). When working with $\Sigma_n$, the base is composed by the \emph{irreducibles} of $\Sigma_n$, a set of functions whose image are invertible matrices. These are formally defined in the context of the representation theory. 

\begin{definition}[Representation]  
$\rho:\Sigma_n\longrightarrow GL_m$ is a representation if it preserves the group product, that is, $\rho(\sigma_1\sigma_2) = \rho(\sigma_1)\cdot\rho(\sigma_2)$ for all $\sigma_1, \sigma_2$. $GL_m$ is the set of invertible complex-valued matrices. The size $m$ of the matrices is also called the dimension of the representation.
\end{definition}

\begin{definition}[Equivalence]
Two representations, $\rho_1$ and $\rho_2$ are equivalent if there exists an invertible matrix $C$ such that
$$\rho_2(\sigma) = C^{-1}\cdot\rho_1(\sigma)\cdot C  \qquad \text{for all } \sigma\in\Sigma_n,$$

which is denoted by $\rho_1 \equiv \rho_2$.
\end{definition}

\begin{definition}[Direct sum]
Given two representations, $\rho_1$ and $\rho_2$, their direct sum is the representation $\rho_1 \oplus \rho_2$ that satisfies
\[
\renewcommand\arraystretch{1.3}
\rho_1\oplus\rho_2(\sigma) =\mleft[
\begin{array}{c|c}
  \rho_1(\sigma) & 0 \\
  \hline
  0 & \rho_2(\sigma)
\end{array}
\mright]
\]

\end{definition}

\begin{definition}[Reducibility]  
A representation $\rho$ is reducible if there exist two representations $\rho_1$ and $\rho_2$ such that $\rho = \rho_1\oplus\rho_2$. Otherwise, it is irreducible.
\end{definition}

For a finite group $G$, the set of irreducible representations (up to equivalence) is finite, implying that one can select a finite number of representations and "build up" the rest from them. For $\Sigma_n$, it has been proved that the number of inequivalent irreducible representations is the number of partitions of $n$.

\begin{definition}[Partition of a number]  
$\lambda = (\lambda_1, \cdots, \lambda_k)$ is a partition of $n$, denoted by $\lambda\vdash n$, if $\sum_{i=1}^k\lambda_i = n$.
\end{definition}

Therefore, if one chooses a set of (inequivalent) irreducible representations of $\Sigma_n$, also called system of irreps, it can be indexed by the partitions of $n$. The specific set or irreducible representations that is used in this paper is the canonical one, which is called Young's Orthogonal Representations (YOR) \cite{kondor2010fourier}. The interesting property of these representations is that the matrices are real-valued and orthogonal. Once a set of irreducibles $\{\rho_{\lambda} : \lambda\vdash n  \}$ is established, we can project the function $f$ onto this base and define, thus, the Fourier transform and Fourier coefficients.

\begin{definition}[Fourier transform]  
The Fourier transform of a function \newline$f:\Sigma_n\longrightarrow \mathbb{R}$ at a representation $\rho$ is defined as $$\hat{f}_\rho = \sum_{\sigma} f(\sigma)\rho(\sigma).$$
\end{definition}

\begin{definition}[Fourier coefficients]  
Given a system of irreps, $\{\rho_{\lambda} : \lambda\vdash n  \}$, the Fourier coefficients of a function $f$ are defined as the collection of the Fourier coefficients at each of the irreducibles of the system
$$\{\hat{f}_{\rho_\lambda} : \lambda\vdash n \}.$$
\end{definition}

Note that, unlike the real case, the coefficients are not simply numerical values, they are stored in matrices, instead. Even so, various of the properties on the real line, such as invertibility (Theorem \ref{theo::invertibility}), linearity and the convolution theorem still hold for the symmetric group. \\

\begin{theorem}[Inverse Fourier transform]\label{theo::invertibility}
\cite{huang2009fourier} Working under YOR, a function $f$ can be computed in terms of its Fourier coefficients according to the following formula
$$f(\sigma) = \dfrac{1}{|\Sigma_n|}\sum_\lambda d_{\rho_\lambda} Tr\left[\hat{f}_{\rho_\lambda}^T \cdot \rho_\lambda(\sigma)  \right],$$ 
where $d_{\rho_\lambda}$ is the dimension of the representation $\rho_\lambda$.
\end{theorem}

\section{Combinatorial optimization problems}\label{subsec::COPs_background}

In the field of permutation-based combinatorial optimization, the aim is to optimize an \emph{objective function} $f:\Sigma_n\longrightarrow \mathbb{R}$, that is,to find 
$$\argmin\limits_{\sigma\in\Sigma_n} f(\sigma),\, \text{ if we are minimizing}$$ 
or $$\argmax\limits_{\sigma\in\Sigma_n} f(\sigma),\, \text{ if we are maximizing.}$$\\

In the next sections, we present the three different problems that we are studying: the Linear Ordering Problem, the Traveling Salesman Problem and the Quadratic Assignment Problem.\\

\subsection{Linear Ordering Problem}

The linear ordering problem (LOP) is given by a square matrix $A = [a_{ij}]$ of size $n$, and consists of maximizing the sum of the upper-diagonal elements of $A$, or equivalently, minimizing the sum of the lower-diagonal elements of $A$, among all the possible reorderings of rows and columns. Although the version of maximization is the one usually found in the literature, in this
paper, we consider the version of minimization, as we will refer to all minimization problems (TSP and QAP). So, the problem can be expressed as 
$$\argmin\limits_{\sigma\in\Sigma_n}\ \sum_{j=1}^{n-1} \sum_{i=j+1}^n\!\! a_{\sigma(i)\sigma(j)}$$

\subsection{Traveling Salesman Problem}

The traveling salesman problem (TSP) consists of finding the shortest route that a traveling salesperson should take to visit a number of cities and then come back to the starting point. 
The distances between the cities are given by a square matrix $D = [d_{ij}]$ of size $n$, where $n$ is the number of cities. The problem can be formulated as
$$\argmin\limits_{\sigma\in\Sigma_n}\ d_{\sigma(n)\sigma(1)} + \sum_{i=1}^{n-1} d_{\sigma(i)\sigma(i+1)}$$

\subsection{Quadratic Assignment Problem}

In the Quadratic Assignment Problem (QAP), a set of $n$ facilities has to be assigned to $n$ locations. The aim is to reduce the cost of the flow between the facilities, which depends on the distances between their locations. If $A=[a_{ij}]$ is the distance matrix, which measures the distance between locations, and $A' = [a'_{ij}]$ is the flow matrix, which measures the flow between facilities, then the QAP can be expressed as
$$\argmin\limits_{\sigma\in\Sigma_n}\ \sum_{i=1}^{n-1} \sum_{j=i+1}^n\!\! a_{\sigma(i)\sigma(j)}\cdot a'_{ij} $$

The LOP and the TSP can be reformulated as particular cases of the QAP, if the values of matrix $A'$ are properly fixed, which is a property that is analyzed in more detail in Appendix \ref{appendix::proofs}.\\

%
%

\section{Characterizing problems in Fourier domain}\label{sec::StructureCOPcoeff}

Finding an accurate characterization of the Fourier coefficients of the COPs is one of the fundamental goals of our research. This means that we are interested in describing the properties that all the objective functions associated with a certain problem, such as the LOP, the TSP or the QAP, have in common. Characterizing the problems in the Fourier space could be interesting in the pursuit of our taxonomization, since we could compare different problems, with disparate definitions, by looking at the similarities and differences between their Fourier coefficients.\\

It has already been proved that the QAP has at most four non-zero Fourier coefficients \cite{kondor2010fourier,rockmore2002fast}. This property is undoubtedly restrictive. The number of Fourier coefficients is the number of partitions of $n$, so it depends on $n$ and, what is more, it grows with $n$, too. This implies that, in the case of the objective function of the QAP, as $n$ increases, there is also a higher proportion of null Fourier coefficients. Yet, not all the functions with 4 non-zero Fourier coefficients must necessarily be QAP functions. Consequently, it would be desirable to find a more exact description of the coefficients, in order to have a bijective correspondence between QAP instances and Fourier coefficients. In such a scenery, we could work interchangeably with flow and distance matrices or in the Fourier domain.\\

In this vein, we find definite conditions that LOP and TSP functions must satisfy (Theorems \ref{theo::LOPcoeff} and \ref{theo::asymmetricTSPcoeff}), and empirical evidence suggests that the coefficients of a QAP should follow certain patterns, beyond the already proved sparsity.\\

%
%

\subsection{Characterization of the LOP}\label{subsec::coeffLOP}


Theorem \ref{theo::LOPcoeff} enumerates the properties that the Fourier coefficients of the objective function of any instance of the LOP must fulfil. Apart from having very few non-zero coefficients, the matrices associated with the non-zero coefficients are always rank one and the proportions among columns remain the same whenever the dimension $n$ is fixed. This is a necessary condition that the Fourier coefficients of an LOP function have to meet and is proved in Appendix \ref{appendix::LOPproof}.\\


\begin{theorem}[Fourier Coefficients of the Linear Ordering Problem]\label{theo::LOPcoeff}
If $f:\Sigma_n\longrightarrow \mathbb{R}$ is the objective function of a Linear Ordering Problem and $\lambda\vdash n$ is a partition, then the Fourier coefficients of $f$ have the following properties
\begin{enumerate}
\item $\hat{f}_\lambda = 0$ if $\lambda\ne (n), (n-1,1), (n-2,1,1)$
\item $\hat{f}_\lambda$ has at most rank one for $\lambda = (n-1,1), (n-2,1,1)$. Having rank one is equivalent to the fact that the matrix columns are proportional. For the mentioned partitions and a fixed dimension $n$, the proportions among the columns of $\hat{f}_\lambda$ are the same for all the instances. 
\end{enumerate}
\end{theorem}

Having found such a precise condition, one may wonder whether the reciprocal is true. If one defines a function $f$ as the inverse of certain Fourier coefficients that satisfy the conditions described in the theorem, will $f$ be an LOP? We carried out a series of experiments for lower dimensions in order to answer to this question.

Assume that $\{\hat{f}_\lambda :\lambda\vdash n\}$ is a family of Fourier coefficients that satisfies the conditions:
\begin{enumerate}
\item $\hat{f}_\lambda = 0$ if $\lambda\ne (n), (n-1,1), (n-2,1,1)$.
\item $\hat{f}_\lambda$ hast at most rank one for $\lambda = (n-1,1), (n-2,1,1)$, and the proportions among columns of $\hat{f}_\lambda$ are equal to those of $\hat{g_\lambda}$, where $g$ is an LOP function such that its coefficients $(n-1,1)$ and $(n-2,1,1)$ are non-zero.
\end{enumerate}

Under these conditions, is $f$ an LOP function? To clarify this question, we designed the following experiment for dimensions $n=3,4,5,6$:

\begin{enumerate}
\item Generate a random LOP function $g$ with non-null $(n-1,1)$ and $(n-2,1,1)$ coefficients. Then, compute $\hat{g}_\lambda$ for $\lambda = (n-1,1), (n-2,1,1)$. 
\item Generate random Fourier coefficients $\hat{f}_\lambda$ following the patterns shown in Theorem \ref{theo::LOPcoeff}. 
\begin{itemize}	
	\item For this purpose, firstly set $f_\lambda = 0$, for $\lambda\ne (n), (n-1,1),(n-2,1,1)$.
	\item Generate a uniformly random number in the interval $(-1,1)$ for $\hat{f}_{(n)}$.
	\item For $\lambda = (n-1,1),(n-2,1,1)$,  build coefficient $\hat{f}_{\lambda}$ in the following way: firstly, generate a random column of size $d_{\rho_\lambda}$ (using a uniform distribution in $(-1,1)$). This will be the first column of the matrix. The other columns of the matrix are proportional to the first, with the same proportions as in $\hat{g}_\lambda$.  
\end{itemize}
\item Define function $f$ as the inverse of the coefficient family $\{\hat{f}_\lambda :\lambda\vdash n\}$ via the FT.
\item Check whether there exists an input matrix for the Linear Ordering Problem such that the objective function associated with that matrix is $f$.  \\  
\end{enumerate}

The details for implementing the last step of the experiment are explained in Appendix \ref{app:isLOP_isTSP}, where the problem is reduced to analyzing the solvability of a certain linear system. We run 10 repetitions of the experiment for each of the dimensions $n = 3,4,5,6$. The result was, invariably, that the inverse function was an LOP.\\

This experiment indicates that, probably, for dimensions $n = 3,4,5,6$, the implication stated in Theorem \ref{theo::LOPcoeff} is actually an equivalence. We hypothesize that this could be generalized for any dimension $n$ or, in other words, that Theorem \ref{theo::LOPcoeff} gathers the properties that exactly characterize the Linear Ordering Problem in the Fourier domain.

%
%

\subsection{Characterization of the TSP}\label{subsec::coeffTSP}

As for the previous case, we prove in Theorem \ref{theo::asymmetricTSPcoeff} that the Fourier coefficients of an asymmetric TSP function have to meet very similar requirements to those of the LOP. This is a necessary condition that the Fourier coefficients of a TSP function have to meet and is proved in Appendix \ref{appendix::TSPproof}.\\

\begin{theorem}[Fourier Coefficients of the Traveling Salesman Problem]\label{theo::asymmetricTSPcoeff}
If $f:\Sigma_n\longrightarrow \mathbb{R}$ is the objective function of a Traveling Salesman Problem and $\lambda\vdash n$ is a partition, then the Fourier coefficients of $f$ have the following properties
\begin{enumerate}
\item $\hat{f}_\lambda = 0$ if $\lambda\ne (n), (n-2,2), (n-2,1,1)$
\item $\hat{f}_\lambda$ hast at most rank one for $\lambda = (n-2,2), (n-2,1,1)$. In addition, for the mentioned partitions and a fixed dimension $n$, the proportions among the columns of $\hat{f}_\lambda$ are the same for all the instances. 
\end{enumerate}
\end{theorem}

As for the reciprocal implication, we have proceeded in a similar fashion to the case of the LOP. The experiment has been designed as follows:

\begin{enumerate}
\item Generate a random TSP function $g$ with non-null $(n-2,2)$ and $(n-2,1,1)$ coefficients. Then, compute $\hat{g}_\lambda$ for $\lambda = (n-2,2), (n-2,1,1)$. 
\item Generate random Fourier coefficients $\hat{f}_\lambda$ following the patterns shown in Theorem \ref{theo::asymmetricTSPcoeff}. 
\begin{itemize}
	\item For this purpose, firstly set $f_\lambda = 0$, for $\lambda\ne (n), (n-2,2),(n-2,1,1)$.
	\item Generate a uniformly random number in the interval $(-1,1)$ for $\hat{f}_{(n)}$.
	\item For $\lambda = (n-2,2),(n-2,1,1)$,  build coefficient $\hat{f}_{\lambda}$ in the following way: firstly, generate a random column of size $d_{\rho_\lambda}$ (using a uniform distribution in $(-1,1)$). This will be the first column of the matrix. The other columns of the matrix are proportional to the first, with the same proportions as in $\hat{g}_\lambda$.  
\end{itemize}
\item Define function $f$ as the inverse of the coefficient family $\{\hat{f}_\lambda :\lambda\vdash n\}$ via the FT.
\item Check whether there exists an input matrix for the Traveling Salesman Problem such that the objective function associated with that matrix is $f$.  \\  
\end{enumerate}

The details for implementing the last step of the experiment are explained in Appendix \ref{app:isLOP_isTSP}. In the same way as with the LOP, we run 10 repetitions of the experiment for each of the dimensions $n = 3,4,5,6$. The result was, invariably, that the inverse function was a TSP. Thus, our earlier conclusions for the LOP (section \ref{subsec::coeffLOP}) hold for this different problem as well. \\

\subsubsection{What happens when the TSP is symmetric?}

The previous analysis has been carried out by considering the general TSP, in which the distance matrix does not have to meet any requirements, such as being symmetric, antisymmetric or euclidean. In this section, we study the TSP in Fourier space, when the distance matrix is symmetric. With this purpose, we design two experiments. The first records the Fourier coefficients of random symmetric TSP functions, in order to discover common patterns. The second experiment works the opposite way, by checking whether random coefficients that follow the observed patterns are symmetric TSP functions or not.\\

The steps of the first experiment are the following:

\begin{enumerate}
\item Generate a symmetric distance matrix $D$ of size $n$, with uniformly random values in the interval $(-1,1)$.
\item Compute the FT of the TSP function with input matrix $D$, and store it. 
\item Repeat the first two steps a 10 times.  \\  
\end{enumerate}

We run this experiment for dimensions $n = 3,4,5,6$, and observe that, for each symmetric TSP function $f$, $\hat{f}_{(n-2,1,1)} = 0$. Matrix $\hat{f}_{(n-2,2)}$ has rank 1, and we recorded the proportions among columns for each of the dimensions $n=3,4,5,6$, which are constant whenever $n$ is fixed, as already proved in Theorem \ref{theo::asymmetricTSPcoeff}.\\

We designed the second experiment:
\begin{enumerate}
\item Generate random Fourier coefficients $\hat{f}_\lambda$ following the patterns observed in the first experiment. 
\begin{itemize}
	\item Firstly set $f_\lambda = 0$, for $\lambda\ne (n), (n-2,2)$.
	\item Generate a uniformly random number in the interval $(-1,1)$ for $\hat{f}_{(n)}$.
	\item Build coefficient $\hat{f}_{(n-2,2)}$ in the following way: firstly, generate a random column of size $d_{\rho_{(n-2,2)}}$ (using a uniform distribution in $(-1,1)$). This will be the first column of the matrix. The other columns of the matrix are proportional to the first, with the proportions that have been recorded in the previous experiment.
\end{itemize}  
\item Define function $f$ as the inverse of the coefficient family $\{\hat{f}_\lambda :\lambda\vdash n\}$ via the FT.
\item Check whether there exists a symmetric input matrix for the Traveling Salesman Problem such that the objective function associated with that matrix is $f$.  \\  
\end{enumerate}

For each of the 10 repetitions of this second experiment, the results remained positive.

%
%

\subsection{Comment on the characterization of the QAP}

Whereas the objective function of an LOP or a TSP is an addition of linear terms, the QAP goes a step further and its objective functions are composed of quadratic terms. Thus, tracing the shape of its Fourier coefficients is not as immediate. We designed a basic experiment in order to gain certain intuition. First, generate a random QAP objective function (with matrices $A$ and $A'$ generated uniformly in the interval $(-1,1)$). Then, compute its FT and record the ranks of its coefficients. The results were the following after 10 repetitions (let us recall that previous studies have demonstrated that the only non-zero coefficients are $(n)$, $(n-1,1)$, $(n-2,2)$, $(n-2,1,1)$):\\

\begin{itemize}
\item Coefficient $(n-1,1)$ is rank 3 (this holds for $n\ge 4$; for $n < 4$, it is full-rank).
\item Coefficients $(n-2,2), (n-2,1,1)$ are rank 1.
\end{itemize}

%
%

\section{Rankings associated with problems}\label{sec::rankings}

As explained in the introduction, instead of merely working with objective functions, we are viewing them as permutations of elements of the search space, by sorting the solutions depending on their objective function values. In order to further simplify our task, we have decided to impose a restriction over the type of functions that we are treating in our initial approach. We consider functions without ties between the objective values of different solutions or, in other words, injective functions. Under this condition, functions are understood as total orderings of the permutations of $\Sigma_n$. This condsideration drastically reduces the dimensionality of the space of functions, since a function $f:\Sigma_n\longrightarrow\mathbb{R}$ is converted to a ranking $\tau\in\Sigma_{n!}$ of elements of the search space, which, in our case, are permutations (so $\Sigma_{n!}$ is the space of all possible rankings of permutations of size $n$). Despite being large, $\Sigma_{n!}$ is still a finite space, contrarily to the space of all real functions over permutations.\\
 

Limiting the size of the space of objective functions enables an easier taxonomization task, which is our main target. With this purpose, a number of questions arise. Consider an optimization problem $\mathcal{P}$, we define the expansion of $P$ at dimension $n$ as the subset of permutations (or rankings) that can be generated with problem $P$ when the search space is $\Sigma_n$ and note it by $E_n(\mathcal{P})$. Then, we consider the following questions:
 
\begin{itemize}
\item Given a problem $\mathcal{P}$, what is the set of rankings that it can generate, $E_n(\mathcal{P})$?
\item Given two problems $\mathcal{P}$ and $\mathcal{Q}$, which is the set of rankings that can be generated by both problems? i.e., how is $E_n(\mathcal{P})\cap E_n(\mathcal{Q})$?
\item Some rankings could be efficiently solved for some algorithms. Therefore knowing the rankings that can be produced by a problem could give us an idea of the goodness of an algorithm for a particular problem.
\end{itemize}

The previous questions have to do with taxonomizing combinatorial optimization problems with an eye in its algorithmic solutions, however other relevant scientific questions can be asked looking at the permutation space:

\begin{itemize}
\item Is there a problem (with a closed form expression for the objective function) able to generate all the possible rankings depending on a polynomial number of parameters?

\item Given two problems $\mathcal{P}$ and $\mathcal{Q}$, assume that $\mathcal{P}$ is defined using a number of parameters $l$ and $\mathcal{Q}$ with $l'$, such that $l > l'$. Lets also assume that $|E_n(\mathcal{P})| < |E_n(\mathcal{Q})|$. Is it possible to reparameterize $\mathcal{P}$ with a lower number of parameters?
\end{itemize}

%
%

\subsection{Linear Ordering Problem}\label{subsec::rankingsLOP}

The Linear Ordering Problem has a property that can be exploited for our analysis. If $\sigma^*$ minimizes the elements below the main diagonal, that is

$$\argmin\limits_{\sigma\in\Sigma_n}\ \sum_{j=1}^{n-1} \sum_{i=j+1}^n\!\! a_{\sigma(i)\sigma(j)},$$
then $\sigma^*$ maximizes those above. Therefore, the highest value of the objective function is reached in the reverse $\sigma^r$ of the optimal (minimal) solution $\sigma^*$, i.e., $\sigma^r(i) = \sigma^*(n-i+1) , \  \forall i = 1,\ldots, n$. This property generalizes to other permutations in the search space: if $\sigma$ gets the $k$-th best objective value then the reverse of $\sigma$ reaches the $k$-th worst value. This property implies that LOP can only produce functions whose associated permutations (rankings) are symmetric, considering the symmetry with respect to a permutation and its reverse. Therefore $E_{n}(LOP) \ne \Sigma_{n!}$. Furthermore, this particular structure allows us to find an upper bound for the number of permutations that can be generated when assuming an injective function (any two different solutions of the search space have different objective function values). This number is:
$$|E_{n!}(LOP)| \le 2^{n!/2}\left (\dfrac{n!}{2}\right )!$$



%
%

\subsection{Rankings in Fourier domain: a little experiment}\label{subsec::rankingsLittleExperiment}

We have previously studied (see section \ref{sec::StructureCOPcoeff}) the characterization of certain COPs in the Fourier domain. As our final target would be to create a taxonomy by analyzing the Fourier Transform of a number of objective functions, and we are viewing them as rankings, it could also be interesting to study the relationship between the Fourier coefficients and the rankings generated by the functions. In this regard, we would like to know what happens to the rankings when certain coefficients are set to 0. More specifically, one could wonder what happens when the last coefficient is 0. Can any ranking of elements of the search space be generated with functions $f$ having the last fourier coefficient, $\hat{f}_{(1,\cdots,1)}$ set to 0? \\

For this purpose, we designed an experiment for a low dimension ($n = 3$), which could be considered as a toy example with the aim of answering to the last question. We wanted to observe what happens when the coefficient indexed by the partition $(1,1,1)$ was set to 0. What kind of rankings would appear under this condition? So, for this purpose, we generated, using the Fourier inversion formula, different functions with random Fourier coefficients for the partition $(2,1)$, but not for $(1,1,1)$, where we forced the coefficient to be null. We also set the coefficient indexed by (1) to 0, because the information encoded in this coefficient is the average of the functions values, so it does not affect the ranking. The coefficients were generated uniformly at random in the interval (-1,1). Then, for each of the functions given by the inversion theorem, we stored its ranking. The pseudocode of the experiment is presented in Algorithm \ref{Alg:Little_experiment}, where the number $reps$ of functions that we generated was 20000000. We designed the experiment to discard all the non-injective rankings, but the result was that no non-injective function appeared. When we stored a ranking, we also kept the record of the number of times that we have obtained that same ranking.\\

\begin{algorithm}
\begin{algorithmic}[]
\caption{Pseudocode of the first experiment}\label{Alg:Little_experiment}
\FOR{$i = 1, \cdots, reps$}
\STATE Choose uniformly random coefficients $\hat{f}_{(2,1)}$ for the partition $(2,1)$
\STATE Set the coefficients $\hat{f}_{(1)}$ and $\hat{f}_{(1,1,1)}$ to 0
\STATE $\hat{f} = \{ \hat{f}_{(1)}, \hat{f}_{(2,1)}, \hat{f}_{(1,1,1)} \}$
\STATE $f = Fourier\_Inversion(\hat{f})$ 
\STATE $r = ranking(f)$
\IF {$is\_injective(r)$}
\STATE $store(r)$
\ENDIF
\ENDFOR
\end{algorithmic}
\end{algorithm}

The number of different rankings that appeared was 360, while, in general, the number of possible rankings of permutations is $(n!)! = 720$, which means that we only obtained half of the rankings. In addition, the ranking that we observed the \emph{least} number of times, appeared 20605 times (so we can assume that the probability of generating the rankings that did not appear is very low, if not 0). \\

We could give a complete list of the rankings that were recorded and the ones that were not, but this list appears to be meaningless in the sense that it is difficult to appreciate, at a bare eye, a common pattern that distinguishes when a ranking should appear or not. However, we were able to find this pattern, which resulted to be directly related to the signature. Let us recall this concept. Any permutation can be expressed as a composition of transpositions (that is, permutations that interchange two elements). The decomposition of a permutation in terms of transpositions is not unique, since the number of transpositions and the type of transpositions can vary. However, for a fixed permutation, there is a common pattern among any possible decomposition: the parity of the number of transpositions is always the same. Taking this into account, the \emph{signature} of a permutation is defined as  

$$sign(\sigma) = \left\{\begin{matrix}
\hspace*{3mm} 1\qquad \text{if $\sigma$ can be decomposed in an even number of transpositions}\\ 
\hspace*{1mm} -1\qquad \text{otherwise} \hspace*{77mm}
\end{matrix}\right. $$

We observed that one can know if a ranking was generated or not by simply computing the signature of the permutations of the ranking. For example, would the following ranking appear? 

\begin{enumerate}
\item{ $[1,2,3]$}
\item{ $[1,3,2]$}
\item{ $[2,1,3]$}
\item{ $[2,3,1]$}
\item{ $[3,1,2]$}
\item{ $[3,2,1]$}
\end{enumerate}

In this ranking, the first permutation would be $[1,2,3]$ and the last one $[3,2,1]$. If we compute the signature of each of the permutations, we obtain

\begin{enumerate}
 \item{ $sign([1,2,3]) = 1$} 
\item{ $sign([1,3,2]) = -1$} 
\item{ $sign([2,1,3]) = -1$} 
\item{ $sign([2,3,1]) = 1$} 
\item{ $sign([3,1,2]) = 1$} 
\item{ $sign([3,2,1]) = -1$}
\end{enumerate}

In the case of this particular pattern of signatures, we observed that, in fact, it appeared and, what is more, that all the rankings sharing this common signature appeared too. This observation can be extended to all the patterns, which means that, if we find a particular pattern in the recorded rankings, then we can find in the list all the other different rankings associated with that same pattern.\\

 Table \ref{tab::signaturePatternsN3} summarizes the patterns of signatures of the rankings that appeared and the ones that were generated in the experiment and the ones that were not. For the sake of readability, a signature of 1 is noted by ``+'' and a signature of -1 by ``-''. So the pattern of the example could be expressed as $[+--++-]$.

\begin{table}\caption{Patterns of signatures of the rankings that were/were not generated during the experiment.}\label{tab::signaturePatternsN3}
\begin{tabular}{ c  c  c }
\hline
\rule{0pt}{4mm} \textbf{Generated} & \hspace{1mm} & \textbf{Not generated}     \vspace*{1mm}\\
\cline{1-1}\cline{3-3}
\\
$[+--++-]$ & & $[+++---]$ \\
$[-+-++-]$ & & $[---+++]$ \\
$[--+++-]$ & & $[+-++--]$ \\
$[-++-+-]$ & & $[-+--++]$ \\
$[-+++--]$ & & $[+-+-+-]$ \\
$[+--+-+]$ & & $[-+-+-+]$ \\
$[+-+--+]$ & & $[++--+-]$ \\
$[-++--+]$ & & $[--++-+]$ \\
$[++---+]$ & & $[++-+--]$ \\
$[+---++]$ & & $[--+-++]$ \\
\\
\hline\\
\end{tabular}
\end{table}

%
%

\subsection{QAP: A restricted method for finding impossible rankings}\label{subsec::rankingsQAP}

In section \ref{subsec::rankingsLOP} it was proven that the number of rankings that an LOP can generate is very limited. Bounding the number of rankings was possible thanks to the symmetry property of the LOP. The QAP does not have such a property and, in addition, its formulation is more complicated. So a more sophisticated analysis was needed. \\ 

We checked, for $n=3$, that the QAP can generate any possible ranking, when varying the flows and distances. In other words, the QAP covers up the wole space of rankings $\Sigma_{3!}$. This naturally leads to the following question: does the same happen for $n = 4,5,6,7,\cdots$?\\

Let us note that, for $n = 4$, every QAP function has one zero Fourier coefficient, the one indexed by $(1,1,1,1)$. In the previous section, it was experimentally shown that there are a number of rankings that either cannot be generated or at least are very difficult to generate when $n = 3$ and the coefficient indexed by $(1,1,1)$ is set to 0. Inspired by these results, we designed a method (see appendix \ref{appendix::impossibleRanking}) that exactly answers to the following question: ``given a specific ranking (of permutations of size $n = 4$), can it be generated with a function that has the Fourier coefficient that is indexed by $(1,1,1,1)$ set to 0?''. Following the idea of the signature patterns of table \ref{tab::signaturePatternsN3}, we input a number of permutations to our method, with very specific signature patterns and recorded the results. Tables \ref{tab:listPossibleRankings} and \ref{tab:listImpossibleRankings} show the rankings that we checked. The former gathers the rankings that can be generated without the information contained in the coefficient $(1,1,1,1)$, while the latter gathers the ones that cannot be generated. Table \ref{tab:SignaturePatterns_dim4} summarizes the results of the aforementioned tables and shows their signature patterns.\\

\begin{table*}
\begin{center}
\caption{List of rankings that we discovered that can be generated having a null $(1,1,1,1)$ coefficient.}\label{tab:listPossibleRankings} 
\begin{tabular}{c  c  c  c  c  c  c }
\\\hline
\\
\multicolumn{7}{c}{RANKINGS THAT CAN BE GENERATED}\\
\\
Rank & \hspace{1cm} & \rule{0pt}{1pt} \textbf{Ranking 1} & \hspace{1cm} & \textbf{Ranking 2} & \hspace{1cm} & \textbf{Ranking 3} \\
\\
\cline{3-3}\cline{5-5}\cline{7-7}
\\
1. & & $[2,3,4,1]$& &$[2,4,3,1]$& &$[4,3,2,1]$\\
2. & & $[1,3,4,2]$& &$[4,2,3,1]$& &$[4,2,3,1]$\\
3. & & $[4,3,2,1]$& &$[1,2,3,4]$& &$[2,3,4,1]$\\
4. & & $[4,1,2,3]$& &$[2,3,4,1]$& &$[4,3,1,2]$\\
5. & & $[1,3,2,4]$& &$[1,3,4,2]$& &$[1,4,3,2]$\\
6. & & $[4,3,1,2]$& &$[4,3,2,1]$& &$[3,1,4,2]$\\
7. & & $[1,4,2,3]$& &$[3,4,2,1]$& &$[2,4,1,3]$\\
8. & & $[2,4,1,3]$& &$[2,1,4,3]$& &$[4,1,2,3]$\\
9. & & $[4,2,3,1]$& &$[1,4,3,2]$& &$[1,2,4,3]$\\
10. & & $[1,2,3,4]$& &$[1,4,2,3]$& &$[3,2,1,4]$\\
11. & & $[2,1,4,3]$& &$[3,1,4,2]$& &$[1,3,2,4]$\\
12. & & $[3,4,1,2]$& &$[4,1,2,3]$& &$[2,1,3,4]$\\
13. & & $[3,1,2,4]$& &$[2,1,3,4]$& &$[3,4,2,1]$\\
14. & & $[4,2,1,3]$& &$[1,3,2,4]$& &$[2,4,3,1]$\\
15. & & $[3,2,4,1]$& &$[2,3,1,4]$& &$[3,2,4,1]$\\
16. & & $[2,4,3,1]$& &$[3,2,4,1]$& &$[3,4,1,2]$\\
17. & & $[3,1,4,2]$& &$[3,2,1,4]$& &$[4,1,3,2]$\\
18. & & $[1,2,4,3]$& &$[4,1,3,2]$& &$[1,3,4,2]$\\
19. & & $[1,4,3,2]$& &$[1,2,4,3]$& &$[4,2,1,3]$\\
20. & & $[2,3,1,4]$& &$[4,3,1,2]$& &$[1,4,2,3]$\\
21. & & $[2,1,3,4]$& &$[3,1,2,4]$& &$[2,1,4,3]$\\
22. & & $[4,1,3,2]$& &$[2,4,1,3]$& &$[2,3,1,4]$\\
23. & & $[3,4,2,1]$& &$[3,4,1,2]$& &$[3,1,2,4]$\\
24. & & $[3,2,1,4]$& &$[4,2,1,3]$& &$[1,2,3,4]$\\

\\
\hline\\
\end{tabular}
\end{center}
\end{table*}

\begin{table*}
\begin{center}
\caption{List of rankings that we discovered that cannot be generated having a null $(1,1,1,1)$ coefficient.}\label{tab:listImpossibleRankings} 
\begin{tabular}{c  c  c  c  c  c  c  c  c}
\\\hline
\\
\multicolumn{9}{c}{RANKINGS THAT CANNOT BE GENERATED}\\
\\
Rank & & \rule{0pt}{1pt} \textbf{Ranking 4} & \hspace{0.1cm} & \textbf{Ranking 5} & \hspace{0.1cm} & \textbf{Ranking 6} & \hspace{0.1cm}& \textbf{Ranking 7}\\
\\
\cline{3-3}\cline{5-5}\cline{7-7}\cline{9-9}\\
1. & & $[3,4,2,1]$& &$[4,2,3,1]$& &$[4,2,3,1]$& &$[3,4,2,1]$\\
2. & & $[4,2,3,1]$& &$[3,4,2,1]$& &$[1,2,3,4]$& &$[4,2,3,1]$\\
3. & & $[2,3,4,1]$& &$[2,3,4,1]$& &$[3,4,2,1]$& &$[2,3,4,1]$\\
4. & & $[4,3,1,2]$& &$[4,3,1,2]$& &$[3,1,2,4]$& &$[4,3,1,2]$\\
5. & & $[1,4,3,2]$& &$[1,4,3,2]$& &$[2,3,4,1]$& &$[1,4,3,2]$\\
6. & & $[3,1,4,2]$& &$[3,1,4,2]$& &$[2,3,1,4]$& &$[3,1,4,2]$\\
7. & & $[2,4,1,3]$& &$[2,4,1,3]$& &$[4,3,1,2]$& &$[2,4,1,3]$\\
8. & & $[4,1,2,3]$& &$[4,1,2,3]$& &$[2,1,4,3]$& &$[4,1,2,3]$\\
9. & & $[1,2,4,3]$& &$[1,2,4,3]$& &$[1,4,3,2]$& &$[1,2,4,3]$\\
10. & & $[3,2,1,4]$& &$[3,2,1,4]$& &$[1,4,2,3]$& &$[3,2,1,4]$\\
11. & & $[1,3,2,4]$& &$[1,3,2,4]$& &$[3,1,4,2]$& &$[1,3,2,4]$\\
12. & & $[2,1,3,4]$& &$[2,1,3,4]$& &$[4,2,1,3]$& &$[4,3,2,1]$\\
13. & & $[4,3,2,1]$& &$[4,3,2,1]$& &$[2,4,1,3]$& &$[2,4,3,1]$\\
14. & & $[2,4,3,1]$& &$[2,4,3,1]$& &$[1,3,4,2]$& &$[3,2,4,1]$\\
15. & & $[3,2,4,1]$& &$[3,2,4,1]$& &$[4,1,2,3]$& &$[3,4,1,2]$\\
16. & & $[3,4,1,2]$& &$[3,4,1,2]$& &$[4,1,3,2]$& &$[4,1,3,2]$\\
17. & & $[4,1,3,2]$& &$[4,1,3,2]$& &$[1,2,4,3]$& &$[1,3,4,2]$\\
18. & & $[1,3,4,2]$& &$[1,3,4,2]$& &$[3,4,1,2]$& &$[4,2,1,3]$\\
19. & & $[4,2,1,3]$& &$[4,2,1,3]$& &$[3,2,1,4]$& &$[1,4,2,3]$\\
20. & & $[1,4,2,3]$& &$[1,4,2,3]$& &$[3,2,4,1]$& &$[2,1,4,3]$\\
21. & & $[2,1,4,3]$& &$[2,1,4,3]$& &$[1,3,2,4]$& &$[2,3,1,4]$\\
22. & & $[2,3,1,4]$& &$[2,3,1,4]$& &$[2,4,3,1]$& &$[3,1,2,4]$\\
23. & & $[3,1,2,4]$& &$[3,1,2,4]$& &$[2,1,3,4]$& &$[2,1,3,4]$\\
24. & & $[1,2,3,4]$& &$[1,2,3,4]$& &$[4,3,2,1]$& &$[1,2,3,4]$\\
\\
\hline\\
\end{tabular}
\end{center}
\end{table*}

\begin{table*}
\begin{center}
\caption{Signature patterns of the rankings of Tables \ref{tab:listPossibleRankings} and \ref{tab:listImpossibleRankings}.}\label{tab:SignaturePatterns_dim4} 
\begin{tabular}{ c  c  c  c  c }
\\\hline
\rule{0pt}{30pt}& & Ranking & \hspace{1cm} & Signature patterns \\
\\
\cline{3-3}\cline{5-5}
\\
\multirow{3}{*}{\footnotesize\rotatebox{90}{\bf Possible}} & & 1 & & $[-++---+--+++++++---+-+--]$\\
& & 2 & & $[+-+-++-+-+----++-+--+-++]$\\
& & 3 & & $[+------------+++++++++++] $\\
\\\hline\\
\multirow{3}{*}{\footnotesize\rotatebox{90}{\bf Impossible}} & & 4 & & $[------------++++++++++++]$\\
& & 5 & & $[------------++++++++++++]$\\
& & 6 & & $[-+-+-+-+-+-+-+-+-+-+-+-+]$\\
& & 7 & & $[-----------+++++++++++-+]$\\
\\
\hline\\
\end{tabular}
\end{center}
\end{table*}

Since the Fourier coefficients of a QAP have not been characterized yet, we can assure that the rankings of Table \ref{tab:listImpossibleRankings} cannot possibly be generated by a QAP. However, regarding the rankings of Table \ref{tab:listPossibleRankings}, it can only be assured that they can be generated with only three Fourier coefficients, but maybe they cannot still be generated by a QAP, since the Fourier coefficients of a QAP should fulfil certain restrictions.

%
%

\section{Conclusions}\label{sec::Conclusion}

%

Broadly speaking, a new working framework, based on Fourier analysis, has been established for the theoretical study of COPs. In this context, the characterizations in the Fourier space of the LOP, the TSP and the QAP have been studied. As a result, a fairly accurate description of the LOP and the TSP has been found, and the initial steps for characterizing the QAP have been taken. As for our second line of research, the idea of representing objective functions as rankings of elements of the search space, we have restrictively bounded the rankings that an LOP can generate. We have also linked the two concepts, the FT and the idea of treating functions as rankings, by observing what type of rankings arise when one of the Fourier coefficients is set to 0. This has led to observing that not any ranking can be generated by a QAP.\\  

Future work may firstly envision reinforcing our findings, by giving a mathematical proof of the reciprocal of the theorems regarding the Fourier coefficients of the LOP and the TSP. Furthermore, we would like to find the exact characterization of the QAP, the particular case of the symmetric TSP and possibly the euclidean TSP. We would also like to study the existence of a more restrictive bound for the number of rankings that the LOP can generate, as well as continue with the research regarding the rankings of the QAP. Apart from expanding our findings, we would also like to incorporate new problems, such as the Permutation Flowshop Scheduling Problem, and consider the rankings of the TSP, which still remain unstudied. As a long-term research goal, the idea of working with partial rankings (that is, including non-injective functions) could also be regarded.

\appendix
\normalsize

\section{Proofs of section \ref{sec::StructureCOPcoeff}}\label{appendix::proofs}

This appendix contains the proofs of Theorems \ref{theo::LOPcoeff}  and \ref{theo::asymmetricTSPcoeff} that characterize the Fourier coefficients of the LOP and the TSP. Part of the mathematical background about the FT and the COPs has already been introduced in Sections \ref{sec::FT_background} and \ref{subsec::COPs_background}, respectively. However, more technical background is needed for the mathematical proof of both theorems. 

\subsection{Further mathematical notation and details on the Fourier Transform}\label{app:Theorems_Background}

The \emph{graph function} associated with a matrix $A$ is defined as
$$f_A(\sigma) = A_{\sigma(n)\sigma(n-1)}$$

Regarding the Fourier Transform, let us recall that the Fourier coefficients are indexed by the partitions of $n$. So the partitions of $n$ are used to index the Fourier coefficients of a function. In addition, for a given partition $\lambda\vdash n$, the rows and columns of $\hat{f}_\lambda$ are indexed by \emph{standard tableaus} of shape $\lambda$ (a standard tableau is a Young tableau in which numbers $1,2,\cdots, n$ are placed exactly once and in increasing order both from left to right and from top to bottom). For example, the standard tableaus of shape $\lambda = (2,2)$ are just the following two: 
$$
 \young(12,34) \qquad \young(13,24)
$$

Apart from Young tableaus, there also exist \emph{Young tabloids}, which, intuitively, are Young tableaus in which the elements in the same row are not sorted. Young tabloids have a slightly different representation:

\begin{center}
\ytableausetup
{boxsize=normal,tabloids}
\ytableaushort{
12, 34
} $\qquad$
\ytableaushort{
13, 24
}
\end{center}

We will assume that, for a given partition $\lambda\vdash n$, its Young tabloids are sorted in a certain way, $t_1, t_2, \cdots$ (to see the particular ordering we refer the reader to the appendix of \cite{huang2009fourier}). \\

The \emph{permutation representation} at partition $\lambda\vdash n$ will be denoted as $\tau_\lambda$, that is,
\begin{equation}\label{eq::definitionPermutationRepresentation}
\tau_\lambda(\sigma)_{ij} = \left\{\begin{matrix}
1\qquad \text{ if } \sigma({t_i})=t_j\\ 
0\hspace{13.5mm} \text{otherwise}
\end{matrix}\right.
\end{equation}
For a detailed definition of $\tau_\lambda$, see \cite{huang2009fourier}.\\

In addition, we have the following decompositions of the permutation representation, in terms of the irreducible representations, also refered to as \emph{irreps}, for a number of partitions:
\begin{align}
\tau_{(n)} \equiv \rho_{(n)}  \nonumber \\
\tau_{(n-1,1)} \equiv \rho_{(n)} \oplus \rho_{(n-1,1)} \nonumber \\
\tau_{(n-2,2)} \equiv \rho_{(n)} \oplus \rho_{(n-1,1)} \oplus \rho_{(n-2,2)}  \label{eq::decompositionsPermutationRepresentation}\\
\tau_{(n-2,1,1)} \equiv \rho_{(n)} \oplus \rho_{(n-1,1)} \oplus \rho_{(n-1,1)} \oplus \rho_{(n-2,2)} \oplus \rho_{(n-2,1,1)} \nonumber
\end{align}

That is, the left-sided representations are equal to the direct sums of the right side when a change of basis is applied. In matrix terminology, this is the same as saying that there exists an invertible matrix $C_\lambda$, such that 
$$ \tau_\lambda(\sigma)  = C_\lambda \left [ \bigoplus\limits_{\mu\ge\lambda} \bigoplus\limits_{i=1}^{k_{\lambda\mu}} \rho_\mu (\sigma) \right ] C_\lambda^{-1}$$

This decomposition implies that the FT of a function $f$ at $\tau_\lambda$ can be computed in terms of the FT at $\rho_\mu$, with $\mu \ge \lambda$:
$$ \hat{f}_{\tau_\lambda} = C_\lambda \left [ \bigoplus\limits_{\mu\ge\lambda} \bigoplus\limits_{i=1}^{k_{\lambda\mu}}  \hat{f}_{\rho_\lambda} \right ] C_\lambda^{-1}$$  

We will denote the Fourier transform of a function $f$ at irreducible $\rho_\lambda$ by $\hat{f}_{\rho_\lambda}$. Even so, whenever we consider that this does not lead to ambiguity, we may simplify this notation as $\hat{f}_\lambda = \hat{f}_{\rho_\lambda}$.\\

\subsection{Characterization of the LOP: Proof}\label{appendix::LOPproof}


The aim of this section is to prove Theorem \ref{theo::LOPcoeff} (stated in section \ref{subsec::coeffLOP}). The steps for proving it can be summarized as follows: 

\begin{enumerate}
\item The problem is reduced by observing that the FT of the objective function is proportional to the product of two factors. One of the factors depends on the particular input matrix, while the other remains always the same for any instance. 
\item The FT of this constant factor is analyzed. 
\item The specific shape of this factor is what gives rise to the rank properties listed in Theorem \ref{theo::LOPcoeff}. \\
\end{enumerate}

%
%

\paragraph{Reduction of the problem} The first step of the proof consists in reducing the problem. In \cite{kondor2010fourier}, Kondor proofs (by means of the convolution theorem) that the Fourier coefficients of a QAP can be expressed in terms of the distance data and the flow data. Since the LOP is a particular case of the QAP, this factorization remains valid for the LOP, too. The mentioned result is the following:

\begin{proposition}[Fourier transform of the QAP]\label{prop::FT_QAP_Kondor}
Given a function $f$ expressed as 
\begin{equation}\label{eq::QAPobjectiveFunction}
f(\sigma) = \sum_{i,j=1}^n A_{\sigma(i)\sigma(j)}A'_{ij}, 
\end{equation}
its Fourier coefficients are zero except for coefficients $(n), (n-1, 1), (n-2,2), (n-2,1,1)$. In addition, the values of these coefficients can be factored in terms of the coefficients of the graph functions of $A$ and $A'$

\begin{equation}\label{eq::coefficients_GraphCorrelation}
\hat{f}_\lambda = \dfrac{1}{(n-2)!}\hat{f_A}_\lambda\cdot\hat{f_{A'}}_\lambda^T
\end{equation}
\end{proposition}

\vspace{2mm}

Note that the LOP is a particular case of the QAP. By setting
$$A'_{ij} = \left\{\begin{matrix}
\hspace*{2mm} 1 \qquad \text{ if } i < j \hspace*{3mm} \\ 
\hspace*{2mm} 0\qquad \text{otherwise} 
\end{matrix}\right. $$
equation (\ref{eq::QAPobjectiveFunction}) becomes the objective function of the LOP. Then,
$$f_{A'}(\sigma) = A'_{\sigma(n)\sigma(n-1)} = \left\{\begin{matrix}
\hspace*{2mm}  1\qquad  \text{ if } \sigma(n) < \sigma(n-1)\\ 
\hspace*{2mm}  0\qquad  \text{otherwise}         \hspace*{14mm}
\end{matrix}\right. $$ 
which is the indicator function over the set $ \{\sigma:\ \sigma(n) < \sigma(n-1)\}$,
\begin{equation}\label{eq::graphFunction_LOP}
f_{A'}(\sigma) = \mathbb{1}_{ \{\sigma:\ \sigma(n) < \sigma(n-1)\} }
\end{equation}

So, for an LOP function, by taking into account equation (\ref{eq::coefficients_GraphCorrelation}), we see that the Fourier transform is proportional to (the proportionality function being $1/(n-2)!$) the product of matrix $[\widehat{f_A}]_\lambda$, whose values depend on the instance, and matrix $[\widehat{f_{A'}}]_\lambda^T$, which remains the same for any instance.\\

This implies that the Fourier coefficients of the objective function of an LOP can always be factored in the product of two matrices multiplied by $1/(n-2)!$. One of the matrices, $[\widehat{f_A}]_\lambda$, depends on the instance, whereas the other one, $[\widehat{f_{A'}}]_\lambda^T$, is fixed and, more specifically, it is the transposed of the transform of the indicator function $\mathbb{1}_{  \{ \sigma:\ \sigma(n) < \sigma(n-1)\} }$. \\

%
%

\paragraph{Analysis of $\widehat{f_{A'}}$}  The next step is to see the specific properties that $[\widehat{f_{A'}}]_\lambda$, the constant factor that the Fourier coefficients of any LOP share, fulfils, in order to deduce how this affects the whole set of LOP functions.\\

In \cite{mania2018kernel}, the author proofs certain properties that the FT of the so-called Kendall kernel satisfies. This kernel is partly related to the function $\widehat{f_{A'}}$, via certain functions $g_{ij}$, which appear in an intermediate step of the proof of the properties of the Kendall kernel. The functions $g_{ij}$ are defined as 
$$g_{ij}(\sigma) = 1-2 \cdot \mathbb{1}_{\{\sigma(i)>\sigma(j)\}}$$

As can be observed, when $i = n-1$ and $j = n$, $g_{ij}$ is a linear transformation of $f_{A'}$. Since the FT is also linear, it is immediate that $[\widehat{g_{ij}}]_\lambda$ and $[\widehat{f_{A'}}]_\lambda$ are proportional except for $\lambda = (n)$. The following proposition is an immediate result of the exposition of \cite{mania2018kernel} when one takes into account the relationship between $[\widehat{g_{ij}}]_\lambda$ and $[\widehat{f_{A'}}]_\lambda$, so the proof has not been included in this paper. To see the details, the reader is refered to the aforementioned paper.

\begin{proposition}[Fourier transform of the indicator function]\label{prop::FT_IndicatorsLOP}
The Fourier transform of $f_{A'}(\sigma) =  \mathbb{1}_{ \{\sigma:\  \sigma(n) < \sigma(n-1)\} }$ satisfies
\begin{enumerate}
	\item $[\widehat{f_{A'}}]_\lambda = 0$ if $\lambda\ne (n), (n-1,1), (n-2,1,1)$,
	\item $[\widehat{f_{A'}}]_\lambda$ has rank one for $\lambda = (n-1,1), (n-2,1,1)$. 
\end{enumerate}
\end{proposition}

\paragraph{Final theorem} Having stated Propositions \ref{prop::FT_QAP_Kondor} and \ref{prop::FT_IndicatorsLOP} above, the proof of Theorem \ref{theo::LOPcoeff}, which characterizes the Fourier coefficients of an LOP, follows naturally.

\begin{proof}[Theorem \ref{prop::FT_IndicatorsLOP}: Fourier coefficients of the LOP]

Taking into account equation (\ref{eq::coefficients_GraphCorrelation}), we already know that the Fourier coefficients of an LOP function are the product of two matrices multiplied by $1/(n-2)!$. One of them ($[\widehat{f_A}]_\lambda$) depends on the values of the input matrix $A$ and the other one ($[\widehat{f_{A'}}]_\lambda^T$), according to Lemma \ref{prop::FT_IndicatorsLOP}, has rank one for $\lambda = (n-1, 1), (n-2,1,1)$ and is 0 for any other partition except for $\lambda = (n)$. A basic result of linear algebra states that 
$$rank(AB) \le \min(rank(A), rank(B))$$

This means that the rank of the product of two matrices is at most the lowest of the ranks of both matrices. Hence, when multiplying matrices $[\widehat{f_A}]_\lambda$ and $[\widehat{f_{A'}}]_\lambda^T$, the resulting rank must be less or equal to 1. \\

Secondly, to prove that the proportions among columns are fixed for a given dimension $n$, observe that if $\lambda = (n-1,1)$ or $\lambda = (n-2,1,1)$, $[\widehat{f_{A'}}]_\lambda$ has rank 1 and so does its transpose $[\widehat{f_{A'}}]_\lambda^T$. This means that all the rows and columns are proportional. Thus, we can write all the columns of $[\widehat{f_{A'}}]_\lambda^T$ proportionally to the first column, which we will be denoting by $v$ and the proportions will be denoted $a_i$ for $i = 2,\cdots, d_{\rho_\lambda}$
$$[\widehat{f_{A'}}]_\lambda^T = \begin{bmatrix}
\mathbf{v} & a_2\mathbf{v}  & \cdots & a_{d_{\rho_\lambda}}\mathbf{v} 
\end{bmatrix}$$ 

So, if 
$$[\widehat{f_A}]_\lambda = \begin{bmatrix}
\mathbf{w}_1^T \\ \mathbf{w}_2^T  \\ \cdots \\ \mathbf{w}^T_{d_{\rho_\lambda}} 
\end{bmatrix}$$  

where $\mathbf{w}_1, \cdots, \mathbf{w}_{d_{\rho_\lambda}}$ are arbitrary column vectors of length $d_{\rho_\lambda}$, then\\

$ \hat{f}_\lambda = \dfrac{1}{(n-2)!}\cdot\hat{f_A}_\lambda\cdot\hat{f_{A'}}_\lambda^T =
\dfrac{1}{(n-2)!} \cdot \begin{bmatrix}
\mathbf{w}_1^T \\ \mathbf{w}_2^T  \\ \cdots \\ \mathbf{w}^T_{d_{\rho_\lambda}}. 
\end{bmatrix} \cdot 
\begin{bmatrix}
\mathbf{v} & a_2\mathbf{v}  & \cdots & a_{d_{\rho_\lambda}}\mathbf{v} 
\end{bmatrix} = $\\

\vspace{6mm}

$ = \dfrac{1}{(n-2)!} \cdot
\begin{bmatrix}
\mathbf{w}_1^T\cdot\mathbf{v} & \mathbf{w}_1^T\cdot a_2\mathbf{v} & \cdots & \mathbf{w}_1^T\cdot a_{d_{\rho_\lambda}}\mathbf{v} \\
 
\mathbf{w}_2^T\cdot\mathbf{v} & \mathbf{w}_2^T\cdot a_2\mathbf{v} & \cdots & \mathbf{w}_2^T\cdot a_{d_{\rho_\lambda}}\mathbf{v}\\
 
\vdots & \vdots &  & \vdots \\ 

\mathbf{w}_{d_{\rho_\lambda}}^T\cdot\mathbf{v} & \mathbf{w}_{d_{\rho_\lambda}}^T\cdot a_2\mathbf{v} & \cdots & \mathbf{w}_{d_{\rho_\lambda}}^T\cdot a_{d_{\rho_\lambda}}\mathbf{v} 
\end{bmatrix}  $\\

\vspace{6mm}

For $\lambda = (n-1,1), (n-2,1,1)$, all the columns of $\hat{f}_\lambda$ are proportional to the first one and its proportions are equal to the ones in $\widehat{f_{A'}}_\lambda^T$.

\end{proof}


\subsection{Characterization of the TSP: Proof}\label{appendix::TSPproof}

The aim of this section is to prove Theorem \ref{theo::asymmetricTSPcoeff} (previously stated in section \ref{subsec::coeffTSP}), following the line of thought of the proof for the LOP. This means that the steps for proving the theorem remain the same: reduction of the problem, analysis of the constant factor $[\widehat{f_{A'}]_\lambda}$ and final proof. However, the prove is not as straightforward, because whereas for the LOP the constant factor had already been analyzed in the literature, we have not found the analogue for the TSP. So, unlike with the LOP, we have carried out the entire analysis of the constant factor.     \\

\paragraph{Reduction of the problem} 

The TSP can be formulated via equation (\ref{eq::QAPobjectiveFunction}) as a particular case of the QAP by setting 

$$A'_{ij} = \left\{\begin{matrix}
\hspace*{2mm} 1\qquad \text{if }  j = i + 1  \ \text{ or }\  (i = n \text{ and } j = 1) \\ 
\hspace*{2mm} 0\qquad \text{otherwise} \hspace*{35mm}
\end{matrix}\right. $$

So the graph correlation function $f_{A'}$ of equation (\ref{eq::coefficients_GraphCorrelation}) takes the following value in the case of the TSP:
\begin{equation}\label{eq::graphFunction_TSP}
f_{A'}(\sigma) = A'_{\sigma(n)\sigma(n-1)} = \left\{\begin{matrix}
\hspace*{2mm} 1\qquad \text{if } \sigma(n-1) = \sigma(n) + 1   \qquad\\ 
\hspace*{7mm}  \text{ or }  \qquad\\
\hspace*{7mm}    (\sigma(n-1) = 1 \text{ and } \sigma(n) = n)     \vspace*{4mm}\\
 
\hspace*{2mm} 0\qquad \text{otherwise} \hspace{25mm}
\end{matrix}\right.  
\end{equation}

Or, equivalently,
$$
f_{A'}(\sigma) = \left\{\begin{matrix}
\hspace*{1mm} 1\qquad \text{if }  \sigma(n-1) = \sigma(n) + 1 \text{ (mod $n$)} \\  
\hspace*{1mm} 0\qquad \text{ otherwise}\hspace{31mm}
\end{matrix}\right.  
$$

\paragraph{Analysis of $\widehat{f_{A'}}$} 

For the sake of clarity, the analysis of $\widehat{f_{A'}}$ has been divided into a number of propositions (Propositions \ref{prop::FT_h_partition(n)}, \ref{prop::FT_h_partition(n-1,1)}, \ref{prop::FT_h_partition(n-2,2)}, \ref{prop::FT_h_partition(n-2,1,1)} and \ref{prop::FT_IndicatorsTSP}) and corollaries (Corollaries \ref{cor::h_(n-2,2)} and \ref{cor::h_(n-2,1,1)}).\\ 

Our first step in the computation of $\widehat{f_{A'}}$ consists in computing $[\widehat{f_{A'}}]_{\rho_{(n)}}$. $\rho_{(n)} = 1$, therefore,

$$[\widehat{f_{A'}}]_{\rho_{(n)}} = \sum_\sigma f_{A'}(\sigma)\cdot \rho_{(n)}(\sigma) = \sum_\sigma f_{A'}(\sigma). $$

According to the definition of $f_{A'}$ in equation (\ref{eq::graphFunction_TSP}), $\sum_{\sigma} f_{A'}(\sigma)$ is the number of permutations $\sigma\in\Sigma_n$ such that $ \sigma(n-1) = \sigma(n) + 1 $ \ or \ ($\sigma(n) = n \text{ and } \sigma(n-1) = 1$). The condition $\sigma(n-1) = \sigma(n) + 1$ \ or \ ($\sigma(n-1) = 1 \text{ and } \sigma(n) = n$) can be understood as mapping two elements ($n-1$ and $n$) to fixed values, while the rest are mapped to arbitray values. The pairs of values to which $n-1$ and $n$ can be mapped are listed:

\begin{align}\label{eq::conditionTSPgraph_unfolded}
\sigma(n) = 1,\ \sigma(n-1) = 2 \nonumber \\
\sigma(n) = 2,\ \sigma(n-1) = 3 \nonumber \\
\vdots \hspace{22mm} \\
\sigma(n) = n-1,\ \sigma(n-1) = n \nonumber \\
\sigma(n) = n,\ \sigma(n-1) = 1 \nonumber 
\end{align}

\vspace{3mm}

There are $n$ possible mappings and, in each of them, the remaining $(n-2)$ elements can have any possible order. This means that the number of permutations that satisfy system (\ref{eq::conditionTSPgraph_unfolded}) is
\begin{equation}\label{eq::FT_trivial_fA_TSP}
[\widehat{f_{A'}}]_{\rho_{(n)}} = n\cdot(n-2)!
\end{equation}

At this point and before proceeding with the exposition, let us note that the FT is linear and that, for any constant function $c$, $\hat{c}_\lambda = 0$ for each $\lambda \ne (n)$. So, when a function $f$ is translated (that is, $f$ is transformed to a function $h = f + c$, being $h$ the new function), its Fourier transform remains the same for any partition except for $\lambda = (n)$ (that is, $\hat{f}_\lambda = h_\lambda$ for $\lambda\ne (n)$). So we could translate $f_{A'}$ and its non-trivial Fourier coefficients would remain the same. \\

Instead of working with $f_{A'}$, the analysis is carried out with a translation of $f_{A'}$, $h = f_{A'} + c$. We adjust our constant $c$, such that $\hat{h}_{(n)} = 0$. Even though this choice may initially seem unclear, it will be elucidated later in Proposition \ref{prop::FT_IndicatorsTSP}. For the time being, probably it suffices to observe that it eventually makes the mathematical reasoning easier. Taking into account equation (\ref{eq::FT_trivial_fA_TSP}), it is immediate to see that the function $h$ with $\hat{h}_{(n)} = 0$ is the following:

\begin{equation}\label{def::h}
h(\sigma) = f_{A'}(\sigma) - \dfrac{1}{n-1} 
\end{equation}

Directly studying $\hat{h}_{\rho_\lambda}$ is not an easy task. Instead, we compute $\hat{h}_{\tau_\lambda}$ and deduce the properties of $\hat{h}_{\rho_\lambda}$ by means of the decompositions of equation (\ref{eq::decompositionsPermutationRepresentation}). The computation of $\hat{h}_{\tau_\lambda}$ is quite straightforward, but a bit tedious, because it requires many steps based on elementary combinatorics.\\


\paragraph{Fourier transform at $\boldsymbol{\tau_{(n)}}$}

From the fact that $\tau_{(n)} = \rho_{(n)} = 1$, it obviously follows Proposition \ref{prop::FT_h_partition(n)}.

\begin{proposition}\label{prop::FT_h_partition(n)}
The FT of $h$ at irreducible $\tau_{(n)}$ is zero.
\end{proposition}


\paragraph{Fourier transform at $\boldsymbol{\tau_{(n-1,1)}}$}


Before computing $\hat{h}_{\tau_{(n-1,1)}}$ in Proposition \ref{prop::FT_h_partition(n-1,1)}, it is convenient to have a look at representation $\tau_{(n-1,1)}$. Its definition is based on the standard tabloids of shape $(n-1, 1)$, which are the following\\

\vspace{3mm}

\ytableausetup
{boxsize=normal,tabloids}
\ytableaushort{
2 3 {\ \cdots} \ n, 1
} \qquad
\ytableaushort{
1 3 {\ \cdots} \ n, 2
} \qquad
\ytableaushort{
1 2 {\ \cdots} \ n, 3
} \qquad $\cdots$ \\

\vspace{7mm}

$\cdots \qquad$ \ytableaushort{
1 2 {\ \cdots} \ {n-1}, n
} \\

\vspace{6mm}

$\tau_{(n-1,1)}$ is given by
$$\tau_{(n-1,1)}(\sigma)_{ij} = \left\{\begin{matrix}
\hspace*{2mm} 1  \qquad     \text{ if } \sigma(\tau_i)=\tau_j\\ 
\hspace*{2mm} 0  \qquad     \text{otherwise} \hspace{5mm}
\end{matrix}\right.$$

\vspace{2mm}

Note that there is a one-to-one correspondence between each of the tabloids of shape $(n-1,1)$ and the element in the second row. So, we can index them by their second-row element. With this ordering, $t_1$ and $t_2$ would be the following
 
$$t_1 =  \ytableaushort{
2 3 {\ \cdots} \ n, 1
} \quad \text{ and } \quad t_2 = \ytableaushort{
1 3 {\ \cdots} \ n, 2
} $$

\hspace{5mm}

the condition $\sigma({t_1})=t_2$ means that we are checking whether\\

\begin{center}
\ytableausetup
{boxsize=9mm}

\ytableaushort{
{\sigma(2)} {\sigma(3)}  {\ \cdots} {\sigma(n)}, {\sigma(1)}
}  \qquad = \qquad 
\ytableaushort{
1 3  {\ \cdots} n, 2
} 
\ytableausetup
{boxsize=normal}
\end{center}

This is equivalent to checking if $\sigma(1) = 2$. So $\tau_{(n-1,1)}$ can also be formulated as

\begin{equation}\label{eq::defTau_(n-1,1)}
\tau_{(n-1,1)}(\sigma)_{ij} = \left\{\begin{matrix}
\hspace*{2mm}	1	\qquad	\text{if } \sigma(i)=j\\ 
\hspace*{2mm}	0	\qquad	\text{otherwise}	\hspace{1mm}
\end{matrix}\right.
\end{equation}

\bigskip

\begin{proposition}\label{prop::FT_h_partition(n-1,1)}
The FT of $h$ at irreducible $\tau_{(n-1,1)}$ is zero.
\end{proposition}

\begin{proof}
The FT of $h$ at representation $\tau_{(n-1,1)}$ is
$$\hat{h}_{\tau_{(n-1,1)}} = \sum_{\sigma}h(\sigma)\cdot \tau_{(n-1,1)}(\sigma) $$

Remember that a representation maps permutations to matrices, and $\tau_{(n-1,1)}(\sigma)$ is an \scalebox{0.9}{$(n\!-\!1)\times(n\!-\!1)$} matrix, so $\hat{h}_{\tau_{(n-1,1)}}$ is an \scalebox{0.9}{$(n\!-\!1)\times(n\!-\!1)$} matrix too. Each element of matrix $\hat{h}_{\tau_{(n-1,1)}}$ is  \\
$$[\hat{h}_{\tau_{(n-1,1)}}]_{ij} = \sum_{\sigma}h(\sigma)\cdot [\tau_{(n-1,1)}(\sigma)]_{ij}  $$

Using the definition of $h$ (equation (\ref{def::h})),

\begin{align}\label{eq::h_FTaux}
[\hat{h}_{\tau_\lambda}]_{ij} 
= \sum_{\sigma}\left (f_{A'}(\sigma) - \dfrac{1}{n-1}\right )\cdot [\tau_\lambda(\sigma)]_{ij} 
\nonumber \\ 
= \sum_\sigma f_{A'}(\sigma)\cdot[\tau_\lambda(\sigma)]_{ij} \ - \  \dfrac{1}{n-1}\sum_\sigma [\tau_\lambda(\sigma)]_{ij}.
\end{align}

Considering the definition of $\tau_{(n-1,1)}$ of equation (\ref{eq::defTau_(n-1,1)}), $\sum_\sigma [\tau_{(n-1,1)}(\sigma)]_{ij}$ is the number of permutations such that $\sigma(i) = j$, that is, $(n-1)!$. Therefore, the second term of the subtraction is
\begin{equation}\label{eq::h_FT_rightTerm}
\dfrac{1}{n-1}\sum_\sigma [\tau_{(n-1,1)}(\sigma)]_{ij} = \dfrac{(n-1)!}{n-1} = (n-2)!
\end{equation}

To compute $\sum_\sigma f_{A'}(\sigma)[\tau_{(n-1,1)}(\sigma)]_{ij}$, notice that for each permutation $\sigma$, $$f_{A'}(\sigma)\cdot [\tau_{(n-1,1)}(\sigma)]_{ij}$$ can take either one of two values. If $f_{A'}(\sigma) = 1$ and $[\tau_{(n-1,1)}(\sigma)]_{ij} = 1$, 
$$f_{A'}(\sigma)\cdot[\tau_{(n-1,1)}(\sigma)]_{ij} = 1$$ 

otherwise it is 0. So 
$$\sum_\sigma f_{A'}(\sigma)\cdot[\tau_{(n-1,1)}(\sigma)]_{ij}$$ 
is the number of permutations $\sigma$ such that $f_{A'}(\sigma) = 1$ and $[\tau_{(n-1,1)}(\sigma)]_{ij} = 1$. A permutation $\sigma$ satisfies  $f_{A'}(\sigma) = 1$ and $[\tau_{(n-1,1)}(\sigma)]_{ij} = 1$ if and only if (see equations (\ref{eq::graphFunction_TSP}) and (\ref{eq::defTau_(n-1,1)}) of $f_{A'}$ and $\tau_{(n-1,1)}$) it satisfies the following system of equations

\begin{equation}\label{eq::eqSystemforSigma_(n-1,1)}
 \left\{\begin{matrix} 
\hspace*{2mm} \sigma(n-1) = \sigma(n) + 1  \quad  \text{ or } \quad  (\sigma(n-1) = 1 \text{ and } \sigma(n) = n) \vspace{3mm}  \\

\hspace*{2mm} \sigma(i) = j   \hspace{73mm}
\end{matrix}\right.  
\end{equation}

\bigskip

The number of permutations that satisfy these conditions depends on the value of index $i$.\\

\begin{itemize}
\item If $i\ne n-1,n$,\\
The possible values of $\sigma$ that satisfy the first condition of system (\ref{eq::eqSystemforSigma_(n-1,1)}) have already been listed in (\ref{eq::conditionTSPgraph_unfolded}). In addition, system (\ref{eq::eqSystemforSigma_(n-1,1)}) imposes $\sigma(i) = j$, so this additional condition makes us discard two of the possibilities listed in (\ref{eq::conditionTSPgraph_unfolded}), because $\sigma(n), \sigma(n-1)\ne j$. Consequently, there are $(n-2)$ possible pairs of values that $\sigma(n)$ and $\sigma(n-1)$ can take. For each of this possibilities, we are fixing three elements: $\sigma(n), \sigma(n-1)$ and $\sigma(i)$. Hence, the number of permutations that satisfy system (\ref{eq::eqSystemforSigma_(n-1,1)}) is 
$$\sum_\sigma f_{A'}(\sigma) \cdot [\tau_{(n-1,1)}(\sigma)]_{ij} = (n-2)\cdot (n-3)! = (n-2)! $$

\item If $i = n$,\\
System (\ref{eq::eqSystemforSigma_(n-1,1)}) is simplified:
\begin{equation}\label{eq::eqSystemforSigma_(n-1,1)_simplified}
 \left\{\begin{matrix} 
\hspace*{2mm} 	\sigma(n-1) = \sigma(n) + 1  \quad  \text{ or } \quad  (\sigma(n-1) = 1 \text{ and } \sigma(n) = n) \vspace{3mm} \\
 
\hspace*{2mm} 	\sigma(n) = j   \hspace{73mm}
\end{matrix}\right.  
\end{equation}

$\sigma(n) = j$ is fixed and so is $\sigma(n-1)$. Then, the number of permutations that satisfy system (\ref{eq::eqSystemforSigma_(n-1,1)_simplified}) is the number of permutations that fix two elements:
$$\sum_\sigma f_{A'}(\sigma) \cdot [\tau_{(n-1,1)}(\sigma)]_{ij} = (n-2)! $$

\item If $i = n-1$,\\
This case is analogous to the previous case, where $i = n$; then, 
$$\sum_\sigma f_{A'}(\sigma) \cdot [\tau_{(n-1,1)}(\sigma)]_{ij} = (n-2)!$$
\end{itemize}

We have just seen that
$$\sum_\sigma f_{A'}(\sigma) \cdot [\tau_{(n-1,1)}(\sigma)]_{ij} = (n-2)!$$ 

Taking (\ref{eq::h_FT_rightTerm}) into account,\\

$[\hat{h}_{\tau_{(n-1,1)}}]_{ij} = \sum_\sigma f_{A'}(\sigma) \cdot [\tau_{(n-1,1)}(\sigma)]_{ij} -  \dfrac{1}{n-1} \sum_\sigma [\tau_{(n-1,1)}(\sigma)]_{ij} = $\\
 
$(n-2)! - (n-2)! = 0. $\\
\end{proof}


\paragraph{Fourier transform at $\boldsymbol{\tau_{(n-2,2)}}$}


A Young tabloid of shape $\lambda = (n-2, 2)$ can be exactly determined by the elements in the second and third row. For example, the following tabloid 

\ytableausetup
{boxsize=9mm}

\begin{center}
\ytableaushort{
1 2 {\ \cdots} \ {n\!-\!2}, {n\!-\!1} n
} 
\end{center}

\ytableausetup
{boxsize=normal}

can be identified by the unordered tuple $\{n-1, n\}$. The permutation representation can be expressed as

\begin{equation}\label{eq::defTau_(n-2,2)}
\tau_{(n-2,2)}(\sigma)_{ij} = \left\{\begin{matrix}
\hspace*{2mm} 	1	\qquad 		\text{if } \sigma(\{i_1, i_2\}) = \{j_1, j_2\}\\ 
\hspace*{2mm} 	0	\qquad 		\text{otherwise}   \hspace{19mm}
\end{matrix}\right.
\end{equation}

where the indices $i = \{i_1, i_2\}$ and $j = \{j_1, j_2\}$ are unordered tuples.

\begin{proposition}\label{prop::FT_h_partition(n-2,2)}
The FT of $h$ at irreducible $\tau_{(n-2,2)}$ can be indexed by unordered tuples and takes the following values depending on the row $i = \{i_1, i_2\}$ and the column $j = \{j_1, j_2\}$:

\begin{itemize}
\item If $\{i_1, i_2\}\cap\{n-1, n\} = \emptyset$,
$$[\hat{h}_{\tau_{(n-2,2)}}]_{ij} = \left\{\begin{matrix}\vspace{3mm}
\hspace*{1mm} \dfrac{2(n-3)!}{n-1}\qquad       \text{if }  |j_1-j_2| = 1 \text{ (mod $n$)}    \hspace{35mm}\\
\hspace*{1mm} -\dfrac{4(n-4)!}{n-1}\qquad       \text{if }   |j_1-j_2| \ne 1 \text{ (mod $n$)}     \hspace{35mm}
\end{matrix}\right.$$

\item If $\{i_1, i_2\}\cap\{n-1, n\} = \{n-1\} \text{ or } \{n\}$,
 
$$[\hat{h}_{\tau_{(n-2,2)}}]_{ij} = \left\{\begin{matrix}\vspace{3mm}
\hspace*{1mm} \dfrac{(3-n)(n-3)!}{n-1}\qquad       \text{if } |j_1-j_2| = 1 \text{ (mod $n$)} \\   
\hspace*{1mm} \dfrac{2(n-3)!}{n-1}\qquad       \text{if } |j_1-j_2| \ne 1  \text{ (mod $n$)}      
\end{matrix}\right.$$

\item If $\{i_1, i_2\}\cap\{n-1, n\} = \{n-1, n\}$,

$$[\hat{h}_{\tau_{(n-2,2)}}]_{ij} = \left\{\begin{matrix}\vspace{3mm}

\hspace*{1mm} \dfrac{(n-3)(n-2)!}{n-1}\qquad       \text{if } |j_1-j_2| = 1 \text{ (mod $n$)}     \hspace{15mm} \\ 
\hspace*{1mm} -\dfrac{2(n-2)!}{n-1}\qquad       \text{if } |j_1-j_2| \ne 1 \text{ (mod $n$)}       \hspace{20mm}\\       
\end{matrix}\right.$$

\end{itemize}

\end{proposition}

\begin{proof}

\bigskip
To compute $\hat{h}_{\tau_{(n-2,2)}}$ as specified by equation (\ref{eq::h_FTaux}), we are going to compute again the two terms of the subtraction separately. \\ 


$\sum_\sigma [\tau_{(n-2,2)}(\sigma)]_{ij}$ is the number of permutations such that $\sigma(i_1) = j_1$ and $\sigma(i_2) = j_2$, or, $\sigma(i_1) = j_2$ and $\sigma(i_2) = j_1$. There are $2 \cdot (n-2)!$ permutations satisfying this condition, then
\begin{equation}\label{eq::h_FT_rightTerm_(n-2,2)}
\dfrac{1}{n-1} \sum_\sigma [\tau_{(n-1,1)}(\sigma)]_{ij} = \dfrac{2\cdot(n-2)!}{n-1}
\end{equation}


To compute $\sum_\sigma f_{A'}(\sigma)[\tau_{(n-2,2)}(\sigma)]_{ij}$ we have to count the number of permutations for which $f_{A'}(\sigma) = 1$ and $[\tau_{(n-2,2)}(\sigma)]_{ij} = 1$, that is, how many permutations satisfy the following system of equations:

\begin{equation}\label{eq::eqSystemforSigma_(n-2,2)}
 \left\{\begin{matrix} 
\hspace*{2mm} 	\sigma(n-1) = \sigma(n) + 1  \ \text{ (mod $n$)}     \vspace{2mm}\\ 
\hspace*{2mm} 	\sigma(\{i_1, i_2\}) = \{j_1, j_2\}  \hspace{14mm}
\end{matrix}\right.  
\end{equation}

The number of permutations that satisfy these equations depends on the values of indices $i$ and $j$. We list the values of $\sum_\sigma f_{A'}(\sigma) \cdot [\tau_{(n-2,2)}(\sigma)]_{ij}$, which depend on certain conditions over the indices. Even though there are 6 distinct cases, we only prove the first and second, since the rest of them are computed similarly.

\begin{itemize}
\item If $\{i_1, i_2\} \cap \{n, n-1\} = \emptyset$,
	\begin{itemize}
	\item If $|j_1 - j_2| = 1 \text{ (mod $n$)} $,\\

	the condition for a permutation $\sigma$ to satisfy $\sigma(n-1) = \sigma(n)+1 \text{ (mod $n$)}$ restricts the pairs of values that $\sigma(n-1)$ and $\sigma(n)$ can hold.  All the $n$ possible pairs of values that $\sigma(n-1)$ and $\sigma(n)$ can take are listed in equation (\ref{eq::conditionTSPgraph_unfolded}). However, the additional condition of system (\ref{eq::eqSystemforSigma_(n-2,2)}), that is $\sigma(\{i_1, i_2\}) = \{j_1, j_2\}$, discards some of these pairs, since we have the restrictions $\sigma(n-1) \ne j_1, j_2$ and $\sigma(n) \ne j_1, j_2$. Assume, without loss of generality, that $j_2 = j_1 + 1\text{ (mod $n$)}$. Then, the discarded pairs are

	$$\sigma(n) = j_1-1  \text{ (mod $n$)},\qquad    \sigma(n-1) = j_1$$
	$$\sigma(n) = j_1,\qquad     \sigma(n-1) = j_2 \text{ (mod $n$)}$$
	$$\sigma(n) = j_2,\qquad     \sigma(n-1) = j_2+1 \text{ (mod $n$)}$$

	So we are left with $(n-3)$ different possible values of $\sigma(n-1)$ and $\sigma(n)$. In addition, $\sigma(\{i_1, i_2\}) = \{j_1, j_2\}$ has also two possibilities, that is $\sigma(i_1) = j_1$ and $\sigma(i_2) = j_2$, or $\sigma(i_1) = j_2$ and $\sigma(i_2) = j_1$. So we have $2 \cdot (n-3)$ possible combinations of values for $\sigma(n-1), \sigma(n), \sigma(i_1)$ and $\sigma(i_2)$. For each combination, the rest of the $(n-4)$ elements can be reordered arbitrarily. This implies that the number of permutations that satisfy system (\ref{eq::eqSystemforSigma_(n-2,2)}) is

	$$\sum_\sigma f_{A'}(\sigma)[\tau_{(n-2,2)}(\sigma)]_{ij} = 2\ (n-3)\ (n-4)!$$

	\item If $|j_1 - j_2| \ne 1 \text{ (mod $n$)}$,\\
	This case is very similar to the previous one, but, since $|j_1 - j_2| \ne 1 \text{ (mod $n$)}$, the pairs of values of $\sigma(n)$ and $\sigma(n-1)$ discarded due to the condition $\sigma(\{i_1, i_2\}) = \{j_1, j_2\}$ are different:

	$$\sigma(n) = j_1-1  \text{ (mod $n$)},\qquad    \sigma(n-1) = j_1$$
	$$\sigma(n) = j_1,\qquad     \sigma(n-1) = j_1+1 \text{ (mod $n$)}$$
	$$\sigma(n) = j_2-1  \text{ (mod $n$)},\qquad    \sigma(n-1) = j_2$$
	$$\sigma(n) = j_2,\qquad     \sigma(n-1) = j_2+1 \text{ (mod $n$)}$$

	Therefore, the number of permutations satisfying system (\ref{eq::eqSystemforSigma_(n-2,2)}) is

	$$\sum_\sigma f_{A'}(\sigma) \cdot [\tau_{(n-2,2)}(\sigma)]_{ij} = 2 \cdot (n-4) \cdot (n-4)!$$
	\end{itemize}

\item If $\{i_1, i_2\} \cap \{n-1, n\} = \{n-1\}$ \ or \ $\{i_1, i_2\} \cap \{n-1, n\} = \{n\}$,
	\begin{itemize}
	\item If $|j_1 - j_2| = 1 \text{ (mod $n$)}$ ,\\
	$$\sum_\sigma f_{A'}(\sigma)[\tau_{(n-2,2)}(\sigma)]_{ij} = (n-3)!$$
	\item If $|j_1 - j_2| \ne 1 \text{ (mod $n$)}$,\\
	$$\sum_\sigma f_{A'}(\sigma)[\tau_{(n-2,2)}(\sigma)]_{ij} = 2\ (n-3)!$$
	\end{itemize}

\item If $\{i_1, i_2\} = \{n-1, n\}$,
	\begin{itemize}
	\item If $|j_1 - j_2| = 1 \text{ (mod $n$)}$,\\
	$$\sum_\sigma f_{A'}(\sigma)[\tau_{(n-2,2)}(\sigma)]_{ij} = (n-2)!$$
	\item If $|j_1 - j_2| \ne 1 \text{ (mod $n$)}$,\\
	$$\sum_\sigma f_{A'}(\sigma)[\tau_{(n-2,2)}(\sigma)]_{ij} = 0$$
	\end{itemize}

\end{itemize}

$\hat{h}_{\tau_{n-2,2}}$ is computed by subtracting the two terms in equation (\ref{eq::h_FTaux}). The first term is $\sum_\sigma f_{A'}(\sigma) \cdot [\tau_{(n-2,2)}(\sigma)]_{ij}$ and the second has been calculated in (\ref{eq::h_FT_rightTerm_(n-2,2)}). The subtraction immediately leads to the statement of our proposition. 

\end{proof}

\begin{corollary}\label{cor::h_(n-2,2)}
$\hat{h}_{\tau_{(n-2,2)}}$ is a rank-one matrix.
\end{corollary}

\begin{proof}
We have already computed the value of $\hat{h}_{\tau_{(n-2,2)}}$ in Proposition \ref{prop::FT_h_partition(n-2,2)}. Rows and columns are indexed by unordered tuples $i = \{i_1, i_2\}$ and $j = \{j_1, j_2\}$, respectively. Notice that there are only two different columns in $\hat{h}_{\tau_{(n-2,2)}}$, depending on whether $|j_1-j_2| = 1 \text{ (mod $n$)}$ or $|j_1-j_2| \ne 1 \text{ (mod $n$)}$. It is easy to check that a column given by $|j_1-j_2| = 1 \text{ (mod $n$)}$ is $2/(3-n)$ times a column given by $|j_1-j_2| \ne 1 \text{ (mod $n$)}$.\\ 

We are going to see it by differentiating row cases. Assume that $j = \{j_1, j_2\}$ is an index such that $|j_1-j_2| = 1 \text{ (mod $n$)}$ and $j' = \{j'_1, j'_2\}$ is such that $|j'_1-j'_2| \ne 1 \text{ (mod $n$)}$.

\begin{itemize}
\item If $\{i_1, i_2\} \cap \{n-1, n\} = \emptyset$,
$$[\hat{h}_{\tau_{(n-2,2)}}]_{ij'} = -\dfrac{4(n-4)!}{n-1} = \dfrac{2}{3-n} \cdot \dfrac{2(n-3)!}{n-1} = \dfrac{2}{3-n} \cdot [\hat{h}_{\tau_{(n-2,2)}}]_{ij}$$

\item If $\{i_1, i_2\} \cap \{n-1, n\} = \{n-1\}$ \ or \ $\{i_1, i_2\} \cap \{n-1, n\} = \{n\}$,
$$[\hat{h}_{\tau_{(n-2,2)}}]_{ij'} = \dfrac{2(n-3)!}{n-1} = \dfrac{2}{3-n} \cdot \dfrac{(3-n)(n-3)!}{n-1} = \dfrac{2}{3-n} \cdot [\hat{h}_{\tau_{(n-2,2)}}]_{ij}$$

\item If $\{i_1, i_2\} = \{n-1, n\}$,
$$[\hat{h}_{\tau_{(n-2,2)}}]_{ij'} = -\dfrac{2(n-2)!}{n-1} = \dfrac{2}{3-n} \cdot \dfrac{(n-3)(n-2)!}{n-1} = \dfrac{2}{3-n} \cdot [\hat{h}_{\tau_{(n-2,2)}}]_{ij}$$
\end{itemize}

This implies that all the columns in $\hat{h}_{\tau_{(n-2,2)}}$ are proportional and, in consequence, $\hat{h}_{\tau_{(n-2,2)}}$ is rank-one.

\end{proof}


\paragraph{Fourier transform at $\boldsymbol{\tau_{(n-2,1,1)}}$}


A Young tabloid of shape $(n-2,1,1)$ has the following representation:
\ytableausetup
{boxsize=9mm}

\begin{center}
\ytableaushort{
1 2 {\ \cdots} \ {n\!-\!2}, {n\!-\!1}, n
} 
\end{center}

\ytableausetup
{boxsize=normal}

It can be exactly identified by the elements in the second and third row, which can be represented, for instance, with the tuple $(n-1, n)$. The permutation representation for partition $\lambda = (n-2,1,1)$ can then be expressed as

\begin{equation}\label{eq::defTau_(n-2,1,1)}
\tau_{(n-2,1,1)}(\sigma)_{ij} = \left\{\begin{matrix}
\hspace*{2mm}	1	\qquad	\text{if } \sigma(i_1)=j_1     \text{ and }     \sigma(i_2)=j_2\\ 
\hspace*{2mm}	0	\qquad	\text{otherwise}   \hspace{27mm}
\end{matrix}\right.
\end{equation}

where the indices are ordered tuples $i = (i_1, i_2)$ and $j = (j_1, j_2)$.\\

\begin{proposition}\label{prop::FT_h_partition(n-2,1,1)}
The FT of $h$ at irreducible $\tau_{(n-2,1,1)}$ can be indexed by ordered tuples and takes the following values depending on the row $i = (i_1, i_2)$ and the column $j = (j_1, j_2)$: 

\begin{itemize}
\item If $\{i_1, i_2\}\cap\{n-1, n\} = \emptyset$,
$$[\hat{h}_{\tau_{(n-2,2)}}]_{ij} = \left\{\begin{matrix}\vspace{3mm}

\hspace*{1mm} \dfrac{(n-3)!}{n-1}        \qquad       \text{if }      \{i_1, i_2\}\cap\{n-1, n\} = \emptyset \ \text{ and } \ |j_1-j_2| = 1 \text{ (mod $n$)}    \hspace{0mm}\\

\hspace*{1mm} -\dfrac{2(n-4)!}{n-1}        \qquad       \text{if }      \{i_1, i_2\}\cap\{n-1, n\} = \emptyset \ \text{ and } \ |j_1-j_2| \ne 1 \text{ (mod $n$)}     \hspace{0mm}\\ 

\end{matrix}\right.$$

\item If $i_1 = n$ and $i_2\ne n-1 $, or $i_1 \ne n$ and $i_2 = n-1 $, or $i_1 = n$ and $i_2 = n-1 $,
$$[\hat{h}_{\tau_{(n-2,2)}}]_{ij} = \left\{\begin{matrix}\vspace{3mm}

\hspace*{1mm}  -\dfrac{(n-2)!}{n-1}        \qquad       \text{if }      i_1 = n, i_2\ne n-1  \ \text{ and } \ j_2 = j_1 + 1 \text{ (mod $n$)}    \hspace{6mm}  \\      \vspace{8mm}

\hspace*{1mm}   \dfrac{(n-3)!}{n-1}       \qquad       \text{if }      i_1 = n, i_2\ne n-1  \ \text{ and } \ j_2 \ne j_1 + 1 \text{ (mod $n$)}    \hspace{5mm} \\      \vspace{3mm}

\hspace*{1mm}  -\dfrac{(n-2)!}{n-1}       \qquad       \text{if }      i_1\ne n, i_2 = n-1 \ \text{ and } \ j_2 = j_1 + 1 \text{ (mod $n$)}     \hspace{6mm} \\ \vspace{8mm}

\hspace*{1mm}  \dfrac{(n-3)!}{n-1}         \qquad       \text{if }      i_1\ne n, i_2 = n-1\ \text{ and } \ j_2 \ne j_1 + 1 \text{ (mod $n$)}       \hspace{5mm}\\ \vspace{3mm}    

\hspace*{1mm}  \dfrac{(n-2)(n-2)!}{n-1}       \qquad       \text{if }      i_1 = n, i_2 = n-1  \ \text{ and } \ j_2 = j_1 + 1 \text{ (mod $n$)}          \hspace{0mm}\\

\hspace*{1mm}   -\dfrac{(n-2)!}{n-1}       \qquad       \text{if }      i_1 = n, i_2 = n-1  \ \text{ and } \ j_2 \ne j_1 + 1 \text{ (mod $n$)}        \hspace{0mm}
\end{matrix}\right.$$

\item If $i_1 = n-1$ and $i_2\ne n $, or $i_1 \ne n-1$ and $i_2 = n$, or $i_1 = n-1$ and $i_2 = n $,
$$[\hat{h}_{\tau_{(n-2,2)}}]_{ij} = \left\{\begin{matrix}\vspace{3mm}

\hspace*{1mm}  -\dfrac{(n-2)!}{n-1}       \qquad       \text{if }      i_1\ne n-1, i_2 = n  \ \text{ and } \ j_1 = j_2 + 1 \text{ (mod $n$)}     \hspace{10mm} \\ \vspace{8mm}

\hspace*{1mm}  \dfrac{(n-3)!}{n-1}        \qquad       \text{if }      i_1\ne n-1, i_2 = n  \ \text{ and } \ j_1 \ne j_2 + 1 \text{ (mod $n$)}       \hspace{9mm}\\ \vspace{3mm}    

\hspace*{1mm}  -\dfrac{(n-2)!}{n-1}        \qquad       \text{if }      i_1 = n-1, i_2\ne n  \ \text{ and } \ j_1 = j_2 + 1 \text{ (mod $n$)}   \hspace{10mm} \\      \vspace{8mm}

\hspace*{1mm}  \dfrac{(n-3)!}{n-1}         \qquad       \text{if }      i_1 = n-1, i_2\ne n  \ \text{ and } \ j_1 \ne j_2 + 1 \text{ (mod $n$)}   \hspace{9mm} \\      \vspace{3mm}

\hspace*{1mm}  \dfrac{(n-2)(n-2)!}{n-1}       \qquad         \text{if }      i_1 = n-1, i_2 = n  \ \text{ and } \ j_1 = j_2 + 1 \text{ (mod $n$)}    \hspace{0mm}\\\vspace{2mm}

\hspace*{1mm}  -\dfrac{(n-2)!}{n-1}       \qquad       \text{if }     i_1 = n-1, i_2 = n  \ \text{ and } \ j_1 \ne j_2 + 1 \text{ (mod $n$)}     \hspace{8mm}\\ 
\end{matrix}\right.$$

\end{itemize}

\end{proposition}

\begin{proof}

We proceed as in the proves of Propositions \ref{prop::FT_h_partition(n-1,1)} and \ref{prop::FT_h_partition(n-2,2)}, by computing $\hat{h}_{\tau_{(n-2,1,1)}}$ with equation (\ref{eq::h_FTaux}). In the subtraction, $\sum_\sigma [\tau_{(n-2,1,1)}(\sigma)]_{ij}$ is the number of permutations such that $\sigma(i_1) = j_1$ and $\sigma(i_2) = j_2$. There are $(n-2)!$ permutations satisfying this condition, then

\begin{equation}\label{eq::h_FT_rightTerm_(n-2,1,1)}
\dfrac{1}{n-1} \sum_\sigma [\tau_{(n-2,1,1)}(\sigma)]_{ij} = \dfrac{(n-2)!}{n-1}
\end{equation}

Regarding $\sum_\sigma f_{A'}(\sigma) \cdot [\tau_{(n-2,1,1)}(\sigma)]_{ij}$, note that $f_{A'}(\sigma) \cdot [\tau_{(n-2,1,1)}(\sigma)]_{ij} = 1$ if and only if $\sigma$ satisfies the following system of equations:

\begin{equation}\label{eq::eqSystemforSigma_(n-2,1,1)}
 \left\{\begin{matrix} 
\hspace*{2mm} \sigma(n-1) = \sigma(n) + 1 \ \text{ (mod $n$)}     \vspace{2mm}\\ 
\hspace*{2mm} \sigma(i_1) = j_1 \vspace{2mm}   \hspace{30mm}\\
\hspace*{2mm} \sigma(i_2) = j_2                \hspace{30mm}
\end{matrix}\right.  
\end{equation}

Similarly to the proves of Propositions \ref{prop::FT_h_partition(n-1,1)} and \ref{prop::FT_h_partition(n-2,2)}, we have to distinguish different cases. Since all the cases can be calculated using basic combinatorics, we will only develop the third and forth as an example (we will skip the first and second ones because they are very similar to those explained in the proof of Proposition \ref{prop::FT_h_partition(n-2,2)}). \\

\begin{itemize}
\item If $\{i_1, i_2\} \cap \{n-1, n\} = \emptyset$,
	\begin{itemize}
	\item If $|j_1 - j_2| = 1 \text{ (mod $n$)}$,\\
	$$\sum_\sigma f_{A'}(\sigma) \cdot [\tau_{(n-2,1,1)}(\sigma)]_{ij} = \dfrac{(n-3)!}{n-1} $$
	\item If $|j_1 - j_2| \ne 1 \text{ (mod $n$)}$,\\
	$$\sum_\sigma f_{A'}(\sigma) \cdot [\tau_{(n-2,1,1)}(\sigma)]_{ij} =  -\dfrac{2(n-4)!}{n-1}$$
	\end{itemize}

\item If $i_1 = n$ and $i_2 \ne n-1$,
	system (\ref{eq::eqSystemforSigma_(n-2,1,1)}) becomes
	\begin{equation}\label{eq::eqSystemforSigma_(n-2,1,1)_i1=n}
	 \left\{\begin{matrix} 
	\hspace*{-4mm} \sigma(n) + 1 = \sigma(n-1) \ \text{ (mod $n$)}     \vspace{2mm}\\ 
	\hspace*{2mm} \sigma(n) = j_1 \vspace{2mm}   \hspace{37mm}\\
	\hspace*{2mm} \sigma(i_2) = j_2                \hspace{37mm}
	\end{matrix}\right.  
	\end{equation}

	\begin{itemize}
	\item If $j_2 = j_1 + 1 \text{ (mod $n$)}$,\\
	System (\ref{eq::eqSystemforSigma_(n-2,1,1)_i1=n}) is incompatible because 
$$\sigma(i_2) = j_2 = j_1+1 = \sigma(n)+1 = \sigma(n-1) \text{ (mod $n$)}$$ Then, $i_2 = n-1$, which is a contradiction. Therefore,

	$$\sum_\sigma f_{A'}(\sigma) \cdot [\tau_{(n-2,1,1)}(\sigma)]_{ij} = 0 $$
	\item Otherwise,\\
	$\sigma(n)$, $\sigma(n-1)$ and $\sigma(i_2)$ are fixed and the rest of the elements can be arbitrarily reordered. So the number of permutations that satisfy system (\ref{eq::eqSystemforSigma_(n-2,1,1)_i1=n}) is the number of permutations that fix 3 elements, that is,
	$$\sum_\sigma f_{A'}(\sigma) \cdot [\tau_{(n-2,1,1)}(\sigma)]_{ij} = (n-3)!$$
	\end{itemize}

\item If $i_1 \ne n-1$ and $i_2 = n$,
	\begin{itemize}
	\item If $j_1 = j_2 + 1 \text{ (mod $n$)}$,\\
	$$\sum_\sigma f_{A'}(\sigma) \cdot [\tau_{(n-2,1,1)}(\sigma)]_{ij} = 0 $$
	\item Otherwise,\\
	$$\sum_\sigma f_{A'}(\sigma) \cdot [\tau_{(n-2,1,1)}(\sigma)]_{ij} =  (n-3)! $$
	\end{itemize}

\item If $i_1 = n-1$ and $i_2 \ne n-1$,
	\begin{itemize}
	\item If $j_1 = j_2 + 1 \text{ (mod $n$)}$,\\
	$$\sum_\sigma f_{A'}(\sigma) \cdot [\tau_{(n-2,1,1)}(\sigma)]_{ij} = 0 $$
	\item Otherwise,\\
	$$\sum_\sigma f_{A'}(\sigma) \cdot [\tau_{(n-2,1,1)}(\sigma)]_{ij} = (n-3)!$$
	\end{itemize}

\item If  $i_1 \ne n$ and $i_2 = n-1$,
	\begin{itemize}
	\item If $j_2 = j_1 + 1 \text{ (mod $n$)}$,\\
	$$\sum_\sigma f_{A'}(\sigma) \cdot [\tau_{(n-2,1,1)}(\sigma)]_{ij} = 0 $$
	\item Otherwise,\\
	$$\sum_\sigma f_{A'}(\sigma) \cdot [\tau_{(n-2,1,1)}(\sigma)]_{ij} = (n-3)!$$
	\end{itemize}

\item If $i_1 = n$ and $i_2 = n-1$,
	\begin{itemize}
	\item If $j_2 = j_1 + 1 \text{ (mod $n$)}$,\\
	$$\sum_\sigma f_{A'}(\sigma) \cdot [\tau_{(n-2,1,1)}(\sigma)]_{ij} = 0 $$
	\item Otherwise,\\
	$$\sum_\sigma f_{A'}(\sigma) \cdot [\tau_{(n-2,1,1)}(\sigma)]_{ij} = (n-3)! $$	
	\end{itemize}

\item If $i_1 = n-1$ and $i_2 = n$,
	\begin{itemize}
	\item If $j_1 = j_2 + 1 \text{ (mod $n$)}$,\\
	$$\sum_\sigma f_{A'}(\sigma) \cdot [\tau_{(n-2,1,1)}(\sigma)]_{ij} = (n-2) $$
	\item Otherwise,\\
	$$\sum_\sigma f_{A'}(\sigma) \cdot [\tau_{(n-2,1,1)}(\sigma)]_{ij} = 0 $$
	\end{itemize}

\end{itemize}

The result stated in the proposition follows from subtracting the value calculated in  (\ref{eq::h_FT_rightTerm_(n-2,1,1)}) to $\sum_\sigma f_{A'}(\sigma) \cdot [\tau_{(n-2,1,1)}(\sigma)]_{ij}$, as indicated by equation (\ref{eq::h_FTaux}).

\end{proof}

\begin{corollary}\label{cor::h_(n-2,1,1)}
The rank of $\hat{h}_{\tau_{(n-2,1,1)}}$ is 2.
\end{corollary}

\begin{proof}
A way of proving that a matrix has rank 2 is by showing that the linear space spanned by its rows has dimension 2. We are going to take this path by defining two row vectors $\mathbf{v}$ and $\mathbf{w}$ which can generate any of the rows of $\hat{h}_{\tau_{(n-2,1,1)}}$. Before defining $\mathbf{v}$ and $\mathbf{w}$, note that, independently of the size of $\hat{h}_{\tau_{(n-2,1,1)}}$, there are only 7 different rows, determined by the following conditions on $i = (i_1, i_2)$:\\

\begin{enumerate}
\item $\{i_1, i_2\} = \emptyset$
\item $i_1 = n-1, \ i_2 \ne n$
\item $i_1 \ne n-1, \ i_2 = n$
\item $i_1 = n-1, \ i_2 = n$
\item $i_1 = n, \ i_2 \ne n-1$
\item $i_1 \ne n, \ i_2 = n-1$
\item $i_1 = n, \ i_2 = n-1$
\end{enumerate}

However, it is easy to see that, actually, there are only 3 types of rows, if we consider proportional rows to be equivalent.\\

\begin{enumerate}
\item Rows satisfying $\{i_1, i_2\} = \emptyset$ are proportional to $\mathbf{z}$ defined as 
$$ \mathbf{z}_j = \left\{\begin{matrix} 
\hspace*{8mm} 1  \qquad\hspace{4mm} \text{ if } |j_1-j_2| = 1 \text{ (mod $n$)}   \vspace{2mm}\\ 
\hspace*{1mm} -\dfrac{2}{n-3}    \qquad \text{ if } |j_1-j_2| \ne 1 \text{ (mod $n$)}
\end{matrix}\right.  $$

\item Rows satisfying $i_1 = n-1$ and $i_2 \ne n$, or $i_1 \ne n-1$ and $i_2 = n$, or $i_1 = n-1$ and $i_2 = n$ are proportional to $\mathbf{v}$ defined as
$$\mathbf{v}_j = \left\{\begin{matrix} 
\hspace*{8mm} 1 \qquad\hspace{4mm} \text{ if } j_1 = j_2 + 1 \text{ (mod $n$)}   \vspace{2mm}\\ 
\hspace*{1mm} -\dfrac{1}{n-2}    \qquad \text{ if } j_1 \ne j_2 + 1 \text{ (mod $n$)}
\end{matrix}\right.  $$

\item Rows satisfying  $i_1 = n$ and $i_2 \ne n-1$, or $i_1 \ne n$ and $i_2 = n-1$, or $i_1 = n$ and $i_2 = n-1$ are propotional to $\mathbf{w}$ defined as
$$\mathbf{w}_j = \left\{\begin{matrix} 
\hspace*{8mm} 1 \qquad\hspace{4mm} \text{ if } j_2 = j_1 + 1 \text{ (mod $n$)}   \vspace{2mm}\\ 
\hspace*{1mm} -\dfrac{1}{n-2}    \qquad \text{ if } j_2 \ne j_1 + 1 \text{ (mod $n$)}
\end{matrix}\right.  $$
\end{enumerate}

It could seem, then, that the rows of $\hat{h}_{\tau_{(n-2,1,1)}}$ are spanned by three vectors, namely $\mathbf{v}$, $\mathbf{w}$ and $\mathbf{z}$; but $\mathbf{z}$ is a linear combination of $\mathbf{v}$ and $\mathbf{w}$. Indeed, 
$$\mathbf{z} = (\mathbf{v} + \mathbf{w})\cdot\dfrac{n-2}{n-3}$$

It follows that $\hat{h}_{\tau_{(n-2,1,1)}}$ is rank-2, since its rows can be expressed as linear combinations of 2 vectors: $\mathbf{v}$ and $\mathbf{w}$.\\
\end{proof}


\begin{proposition}[Fourier transform of the indicator $f_{A'}$]\label{prop::FT_IndicatorsTSP}
The Fourier transform of $f_{A'}$ as defined by equation (\ref{eq::graphFunction_TSP}) satisfies 
\begin{enumerate}
	\item $[\widehat{f_{A'}}]_{\rho_{(n)}} \ne 0$,
	\item $[\widehat{f_{A'}}]_{\rho_{(n-1,1)}} = 0$
	\item $[\widehat{f_{A'}}]_{\rho_\lambda}$ has rank one for $\lambda = (n-2,2), (n-2,1,1)$. 
\end{enumerate}
\end{proposition}

\vspace{2mm}

\begin{proof}

\begin{enumerate}

\item We have seen this in Proposition \ref{prop::FT_h_partition(n)}.\\

\item We know that $\hat{h}_{\tau_\lambda} = 0$ for $\lambda = (n)$ (Proposition \ref{prop::FT_h_partition(n)}) and for $\lambda = (n-1,1)$ (Proposition \ref{prop::FT_h_partition(n-1,1)}). We can see, thanks to the decompositions of the permutation representations of equation (\ref{eq::decompositionsPermutationRepresentation}), that $\hat{h}_{\tau_{(n-1,1)}}$ is equivalent to the direct sum of $\hat{h}_{\rho_{(n)}}$ and $\hat{h}_{\rho_{(n-1,1)}}$. Since $\hat{h}_{\tau_{(n-1,1)}} = 0$ and $\hat{h}_{\rho_{(n)}} = 0$, we conclude that $\hat{h}_{\rho_{(n-1,1)}} = 0$.\\

\item In the previous step, we have seen that $\hat{h}_{\rho_{(n)}} = 0$ and $\hat{h}_{\rho_{(n-1,1)}} = 0$, then the decomposition of $\hat{h}_{\tau_{(n-2,2)}}$ is reduced to:

$$\hat{h}_{\tau_{(n-2,2)}} \equiv \hat{h}_{\rho_{(n)}} \oplus \hat{h}_{\rho_{(n-1,1)}} \oplus \hat{h}_{\rho_{(n-2,2)}}  \  \Longleftrightarrow  \  \hat{h}_{\tau_{(n-2,2)}} \equiv \hat{h}_{\rho_{(n-2,2)}}$$

We have proved in Corollary \ref{cor::h_(n-2,2)} that $\hat{h}_{\tau_{(n-2,2)}}$ is rank-one, then, since $\hat{h}_{\rho_{(n-2,2)}}$ is equivalent to $\hat{h}_{\tau_{(n-2,2)}}$, it has to be rank-one too.\\

The only non-zero components in the decomposition of $\hat{h}_{\tau_{(n-2,1,1)}}$ in terms of the irreducible representations are $\hat{h}_{\rho_{(n-2,2)}}$ and $\hat{h}_{\rho_{(n-2,1,1)}}$, that is

$$\hat{h}_{\tau_{(n-2,1,1)}} \equiv \hat{h}_{\rho_{(n-2,2)}} \oplus \hat{h}_{\rho_{(n-2,1,1)}}$$

 $\hat{h}_{\rho_{(n-2,2)}}$ has rank 1 and $\hat{h}_{\tau_{(n-2,1,1)}}$ rank 2 (see Corollary \ref{cor::h_(n-2,1,1)}). Then,$\hat{h}_{\rho_{(n-2,1,1)}}$ must have rank one too.\\   

$[\widehat{f_{A'}}]_{\rho_\lambda} = \hat{h}_{\rho_\lambda}$, except for $\lambda = (n)$, which means that their rank properties are the same, for any coefficient $\lambda\ne (n)$. This concludes our proof.\\
\end{enumerate}

\end{proof}

\paragraph{Final theorem} 

The proof of the theorem to which was devoted this whole section (Theorem \ref{theo::asymmetricTSPcoeff} in section \ref{subsec::coeffTSP}) is analogous to the proof Theorem \ref{theo::LOPcoeff}, by taking into account Propositions \ref{prop::FT_QAP_Kondor} and \ref{prop::FT_IndicatorsTSP}.

\section{{\ttfamily isLOP} and {\ttfamily isTSP} functions}\label{app:isLOP_isTSP}


\subsection{{\ttfamily isLOP} function}\label{subsec:isLOP}

The aim of {\ttfamily isLOP} is to check whether a function $f:\Sigma_n\rightarrow\mathbb{R}$ corresponds to an LOP or not. Given a set of objective values, it specifically checks whether there exists a matrix $A$ such that $f$ is the objective function of an LOP. Assume that there exists a certain ordering among the permutations of size $n$, which means that they are indexed from 1 to $n!$. Then, we can write $\Sigma_n = \{\sigma_1, \sigma_2, \sigma_3, \cdots, \sigma_{n!} \}$. Our question can specifically be phrased as follows: given certain objective values $v_1, v_2, \cdots, v_{n!}$, does a matrix $A = [a_{ij}]\in\mathbb{R}^{n\times n}$ exist such that the LOP function $f$ obtained with input matrix $A$ satisfies $f(\sigma_l) = v_l$, for $l=1,\cdots, n!$? This is the same as wondering whether there exists $A$ such that
\begin{equation}\label{eq:isLOP_initialEquations}
v_l = f(\sigma_l) =  \sum_{i = 1}^{n-1} \sum_{j = i+1}^{n} a_{\sigma_l(i) \sigma_l(j)}, \qquad l = 1,2,\cdots, n!
\end{equation}

Equation (\ref{eq:isLOP_initialEquations}) can be expressed as a linear system, by performing a few operations: 

\begin{equation}\label{eq:isLOP_auxEquations}
 v_l =  \sum_{i = 1}^{n-1} \sum_{j = i+1}^{n} a_{\sigma_l(i) \sigma_l(j)} =  \sum_{s = 1}^{n} \sum_{t = 1}^{n} m_{st}(\sigma_l) a_{st},
\end{equation}

with $m_{st}(\sigma_l)= \left\{\begin{matrix}
\hspace*{2mm}	1 	\qquad	\text{ if } \sigma_l^{-1}(s) < \sigma_l^{-1}(t) \\ 
\\
\hspace*{2mm}	0 	\qquad	\text{ otherwise} \hspace{13mm}
\\ 
\end{matrix}\right.$ \\

Equation (\ref{eq:isLOP_auxEquations}) can be further transformed by mapping the double indices $st$ to a single index $r$, by using the following relation:
$$r = t + (s-1) \cdot n$$

So, by setting $\tilde{a}_r = a_{st}$ and $\tilde{m}_{lr} = m_{st}(\sigma_l)$, equation (\ref{eq:isLOP_auxEquations}) can be rewritten as a linear system:

\begin{equation}\label{eq:isLOP_finalEquations}
 v_l = \sum_{r = 1}^{n^2} \tilde{m}_{lr} \tilde{a}_r,
\end{equation}

with 

\begin{equation}\label{eq:isLOP_matrixElements}
\tilde{m}_{lr}= \left\{\begin{matrix}
\hspace*{2mm}	1	\qquad	\text{ if } \sigma_l^{-1}(s) < \sigma_l^{-1}(t) \\ 
\\
\hspace*{2mm}	0	\qquad	\text{ otherwise}	\hspace{13mm}
\\ 
\end{matrix}\right.
\end{equation}

Since $\tilde{a}_1, \tilde{a}_2, \cdots, \tilde{a}_{n^2}$ are the unknown variables, by defining $\tilde{M} = [\tilde{m}_{lr}]$ and $\mathbf{v} = [v_1 v_2 \cdots v_{n!}]^T$, one can know if the system defined by equation (\ref{eq:isLOP_initialEquations}) has a solution (that is, if the given values $v_1, v_2, \cdots, v_{n!}$ are the objective values of an LOP) by knowing if the following linear system is solvable:
$$\tilde{M}\mathbf{x} = \mathbf{v}$$

Note that $\tilde{M}\in\mathbb{R}^{n!\times n^2}$ is not square. A way of tackling this problem is by finding the least-squares solution. Thus, we solve the problem
$$\min_{\mathbf{x}_\in\mathbb{R}^{n^2} }|| \tilde{M}\mathbf{x} - \mathbf{v} ||$$
If there exists $\mathbf{x}$ such that the norm is 0, then we have found the coefficients of the input matrix of the LOP, and the answer is positive. This works theoretically, but, since the problem is solved computationally, one has to establish a threshold to check whether the norm is approximately 0. If the norm definitely is non-zero, $v_1,v_2,\cdots, v_{n!}$ cannot be the objective-function values of an LOP. Algorithm \ref{Alg:isLOP} summarizes the procedure. A given ordering among permutations is assumed (we used the one given by SnFFT julia package).\\

\begin{algorithm}
\hspace*{\algorithmicindent} \textbf{Input:} $v_1, v_2, \cdots, v_{n!}$ \\
\hspace*{\algorithmicindent} \textbf{Output:} $isLOP$ 

\begin{algorithmic}[]
\caption{Pseudocode of {\ttfamily isLOP}}\label{Alg:isLOP}
\STATE Build matrix $\tilde{M} = [\tilde{m}_{lr}]$, as described by Eq. (\ref{eq:isLOP_matrixElements})  
\STATE Solve $r = \min_{\mathbf{x}_\in\mathbb{R}^{n^2} }|| \tilde{M}\mathbf{x} - \mathbf{v} ||$
\IF {$r = 0$}
\STATE $isLOP = true$
\ELSE
\STATE $isLOP = false$
\ENDIF
\RETURN $isLOP$
\end{algorithmic}
\end{algorithm}


\subsection{{\ttfamily isTSP} function}\label{subsec:isTSP}

{\ttfamily isTSP} is the twin function of {\ttfamily isLOP}, and it checks whether a function $f:\Sigma_n\rightarrow\mathbb{R}$ corresponds to a TSP or not. Assume again that there exists a certain ordering among the permutations of size $n$, so $\Sigma_n = \{\sigma_1, \sigma_2, \sigma_3, \cdots, \sigma_{n!} \}$. Then, our question can specifically be phrased as follows: given certain values $v_1, v_2, \cdots, v_{n!}$, does a matrix $D = [d_{ij}]\in\mathbb{R}^{n\times n}$ exist such that the TSP function $f$ obtained with input matrix $D$ satisfies $f(\sigma_l) = v_l$, for $l=1,\cdots, n!$? This is the same as wondering whether there exists $D$ such that
\begin{equation}\label{eq:isTSP_initialEquations}
v_l = f(\sigma_l) =  d_{\sigma_l(n)\sigma_l(1)}\sum_{i = 1}^{n-1} d_{\sigma_l(i) \sigma_l(i+1)}, \qquad l = 1,2,\cdots, n!
\end{equation}

Similarly to the case of the LOP, equation (\ref{eq:isTSP_initialEquations}) can be expressed as a linear system, by performing a few operations: 

\begin{equation}\label{eq:isTSP_auxEquations}
 v_l =  d_{\sigma_l(n)\sigma_l(1)}\sum_{i = 1}^{n-1} d_{\sigma_l(i) \sigma_l(i+1)} = \sum_{s = 1}^{n} \sum_{t = 1}^{n} m_{st}(\sigma_l) d_{st},
\end{equation}

with $m_{st}(\sigma_l)= \left\{\begin{matrix}
\hspace*{2mm}	1	\qquad	\text{ if } \sigma^{-1}(s) + 1 = \sigma^{-1}(t) \quad (mod\ n) \\ 
\\
\hspace*{2mm}	0	\qquad	\text{ otherwise}	\hspace{35mm}
\\ 
\end{matrix}\right.$ \\

Equation (\ref{eq:isTSP_auxEquations}) can be further transformed by mapping the double indices $st$ to a single index $r$, by using the same relation as for {\ttfamily isLOP}:
$$r = t + (s-1) \cdot n$$

So, by setting $\tilde{d}_r = d_{st}$ and $\tilde{m}_{lr} = m_{st}(\sigma_l)$, equation (\ref{eq:isTSP_auxEquations}) can be rewritten as a linear system:

\begin{equation}\label{eq:isTSP_finalEquations}
 v_l = \sum_{r = 1}^{n^2} \tilde{m}_{lr} \tilde{d}_r,
\end{equation}

with 

\begin{equation}\label{eq:isTSP_matrixElements}
\tilde{m}_{lr}= \left\{\begin{matrix}
\hspace*{2mm}	1	\qquad	\text{ if } \sigma^{-1}(s) + 1 = \sigma^{-1}(t) \quad (mod\ n)\\ 
\\
\hspace*{2mm}	0	\qquad	\text{ otherwise}	\hspace*{35mm}	
\\ 
\end{matrix}\right.
\end{equation}

Since $\tilde{d}_1, \tilde{d}_2, \cdots, \tilde{d}_{n^2}$ are the unknown variables, by defining $\tilde{M} = [\tilde{m}_{lr}]$ and $\mathbf{v} = [v_1 v_2 \cdots v_{n!}]^T$, one can know if system (\ref{eq:isTSP_initialEquations}) has a solution (that is, if the given values $v_1, v_2, \cdots, v_{n!}$ are the objective function values of an TSP) by knowing if the following linear system is solvable:
$$\tilde{M}\mathbf{x} = \mathbf{v}$$

Note that $\tilde{M}\in\mathbb{R}^{n!\times n^2}$ is not square. A way of tackling this problem is by finding the least-squares solution. Thus, we would like to solve the problem
$$\min_{\mathbf{x}_\in\mathbb{R}^{n^2} }|| \tilde{M}\mathbf{x} - \mathbf{v} ||$$

However, this cannot be directly solved, unlike the case of the LOP (see section \ref{subsec:isLOP}), because matrix $\tilde{M}$ is never full-rank. This happens because there are many permutations that represent the same solution, e.g, for $n = 4$, $[1,2,3,4]$ and $[3,4,1,2]$. So, there are many rows of $\tilde{M}$ that are repeated. In order to solve the least-squares problem, the repeated rows have to be removed. Before removing them, however, it is necessary to check if permutations representing the same solution share the same objective value. If this does not happen, then we can assure, without solving the least squares, that $v_1, v_2, \cdots, v_{n!}$ cannot be generated by a TSP. In this intermediate step, one has to take into account whether we are considering the symmetric or the asymmetric TSP, because the number of equivalent permutations depends on the case. After having removed the redundant rows, if the objective values are consistent (that is, if equivalent solutions share the same objective values), then the least squares is solved on the reduced matrix. \\

Algorithms \ref{Alg:isConsistent} and \ref{Alg:isTSP} summarize the procedure. A given ordering among permutations is assumed again (the one given by SnFFT julia package). In Algorithm \ref{Alg:isTSP}, $representative(\sigma)$ is a function that computes a representative of the equivalence class of $\sigma$. That is, among all the permutations that encode the same solution as $\sigma$, a single one is chosen to represent the whole group. Note that this function varies depending on whether we are working with the symmetric or the asymmetric version of the TSP and, what is more, this is the only part of {\ttfamily isTSP} that differs between the symmetric and the asymmetic cases. \\

\begin{algorithm}
\hspace*{\algorithmicindent} \textbf{Input:} $v_1, v_2, \cdots, v_{n!}$ \\
\hspace*{\algorithmicindent} \textbf{Output:} $isConsistent$ 

\begin{algorithmic}[]
\caption{Pseudocode of {\ttfamily isCONSISTENT}}\label{Alg:isConsistent}
\STATE isConsistent = true
\FOR{$i = 1, \cdots, n!$}
	\STATE $\sigma_j = representative(\sigma_i)$
	\IF {$v_i != v_j$}
		\STATE $isConsistent = false$
		\STATE \textbf{break}
	\ENDIF
\ENDFOR
\RETURN $isConsistent$
\end{algorithmic}
\end{algorithm}

\begin{algorithm}
\hspace*{\algorithmicindent} \textbf{Input:} $v_1, v_2, \cdots, v_{n!}$ \\
\hspace*{\algorithmicindent} \textbf{Output:} $isTSP$ 

\begin{algorithmic}[]

\caption{Pseudocode of {\ttfamily isTSP}}\label{Alg:isTSP}
\STATE $isTSP = false$
\STATE isConsistent = \textit{isCONSISTENT}$(v_1, v_2, \cdots, v_{n!})$ 
\IF {isConsistent}
	\STATE Build matrix $\tilde{M} = [\tilde{m}_{lr}]$, as described by Eq. (\ref{eq:isTSP_matrixElements})
	\FOR {$i = n!, \cdots, 1$}
		\IF {$\sigma_i$ != $representative(\sigma_i)$}
			\STATE Delete element $i$ from $v_i$ and row $i$ from $\tilde{M}$			
		\ENDIF
	\ENDFOR 
	\STATE Solve $r = \min_{\mathbf{x}_\in\mathbb{R}^{n^2} }|| \tilde{M}\mathbf{x} - \mathbf{v} ||$
	\IF {$r = 0$}
		\STATE $isTSP = true$
	\ENDIF
\ENDIF
\RETURN $isTSP$
\end{algorithmic}
\end{algorithm}


\section{Impossible ranking method}\label{appendix::impossibleRanking}


The method presented in this section of the appendix is inspired in the results of the experiment of Section \ref{subsec::rankingsLittleExperiment}. In this section, an experiment was designed to computationally analyze what type of rankings generates a function $f:\Sigma_3\longrightarrow \mathbb{R}$ with $\hat{f}_{(1,1,1)} = 0$ (that is, its ``last'' Fourier coefficient is set to 0). The results suggested that the rankings that are generated by such functions can be univoquely determined by certain patterns in which the signature is involved, while there exist a number of rankings that cannot be generated when $\hat{f}_{(1,1,1)} = 0$. The realization of the experiment was possible because the size of the search space was low, $|\Sigma_3| = 3! = 6$, and the number of possible rankings was still manageable, $|\Sigma_{3!}| = |\Sigma_6| = 6! = 720$. It would be desirable to extend the results to higher dimensions, but as a matter of fact, for $n = 4$, $|\Sigma_4| = 24$ and $|\Sigma_{4!}| = 24! \approx 6 \cdot 10^{23}$. So extending the previous result is problematic even for $n = 4$. To overcome this limitation, we reformulated our research question. Instead of wondering ``which rankings are generated when $\hat{f}_{(1,1,\cdots,1)} = 1$ is 0?'', one could wonder ``given a specific ranking of size $n!$, can it be generated with a function $f$ such that $\hat{f}_{(1,1,\cdots,1)} = 0$?''. A method that answers to this question when $n = 4$ is presented, and that could eventually be extended for other purposes.
\\


\subsection{Impossible ranking method for $n = 4$}

The aim of this section is to develope a tool that can be used as a black box which answers to the following question: ``given a ranking of permutations of size 4, can it be generated with a function $f$ such that $\hat{f}_{(1,1,1,1)} = 0$?''. A ranking of the elements of $\Sigma_4$ is a permutation $\tau$ in itself (it reorders the elements of $\Sigma_4$), that is, $\tau\in\Sigma_{4!}$. In this context, we are wondering about the following situation: Given an arbitrary permutation $\tau\in\Sigma_{4!}$, is there any function $f:\Sigma_4\to\mathbb{R}$, with $\hat{f}_{\rho_{(1,1,1,1)}}=0$ that could generate the ranking $\tau$?\\




\begin{enumerate}
\item If $h = f + c\cdot g$, then, for any $\lambda$, $\hat{h}_\lambda = \hat{f}_\lambda + c\cdot \hat{g}_\lambda$.
\item If, for any $\lambda$, $\hat{h}_\lambda = \hat{f}_\lambda + c\cdot \hat{g}_\lambda$, then $h = f + c\cdot g$.
\end{enumerate}


The second observation can be useful for decomposing the inverse FT of a family of coefficients. Let us start with a reduced case, $n = 3$. If $n = 3$, the possible partitions of $n$ are 3: $(3)$, $(2,1)$ y $(1,1,1)$. The dimensions of the irreducible representations associated to them (that is, the dimensions of each of the matrix Fourier coefficients) are the following: $d_{(3)} = 1$, $d_{(2,1)} = 2$ and $d_{(1,1,1)} = 1$. Consider the collections of coefficients shown in Table \ref{BaseFamilies_n3} and the functions given by computing their Fourier inversion (we denote by $\mathcal{F}^{-1}$ the inverse FT): $f_0 = \mathcal{F}^{-1}(\hat{f}_0),\ f_1 = \mathcal{F}^{-1}(\hat{f}_1), \cdots,\ f_5 = \mathcal{F}^{-1}(\hat{f}_5)$. Suppose we have arbitrary Fourier coefficients:

\begin{table*}
\centering
\caption{\footnotesize The ``base'' families of coefficients $\hat{f}_0, \hat{f}_1, \hat{f}_2, \hat{f}_3, \hat{f}_4, \hat{f}_5$.}\label{BaseFamilies_n3} 

\bigskip
\begin{tabular}{ c r c r c r c}
  & & $(3)$ & & $(2,1)$  & & $(1,1,1)$\\
\cline{3-3}\cline{5-5}\cline{7-7}
\vspace{-0.1cm} & & & & & & \\
 \vspace{-0.3cm} & & & & & &\\ 
 $\hat{f}_0$:\hspace{0.1cm} &  & \hspace{0.1cm} [1] &  & $\begin{bmatrix} 0 & 0 \\  0 & 0 \end{bmatrix}$  &  \hspace{0.1cm} & [0]\hspace{0.1cm} \\
& & & & & &\\
\hline\\

 $\hat{f}_1$:\hspace{0.1cm} &  & \hspace{0.1cm} [0] &  & $\begin{bmatrix} 1 & 0 \\  0 & 0 \end{bmatrix}$  &  \hspace{0.1cm} & [0]\hspace{0.1cm} \\
& & & & & &\\
\hline\\

 $\hat{f}_2$:\hspace{0.1cm} &  & \hspace{0.1cm} [0] &  & $\begin{bmatrix} 0 & 1 \\  0 & 0 \end{bmatrix}$  &  \hspace{0.1cm} & [0]\hspace{0.1cm} \\
& & & & & &\\
\hline\\

 $\hat{f}_3$:\hspace{0.1cm} &  & \hspace{0.1cm} [0] &  & $\begin{bmatrix} 0 & 0 \\  1 & 0 \end{bmatrix}$  &  \hspace{0.1cm} & [0]\hspace{0.1cm} \\
& & & & & &\\
\hline\\

 $\hat{f}_4$:\hspace{0.1cm} &  & \hspace{0.1cm} [0] &  & $\begin{bmatrix} 0 & 0 \\  0 & 1 \end{bmatrix}$  &  \hspace{0.1cm} & [0]\hspace{0.1cm} \\
& & & & & &\\
\hline\\

 $\hat{f}_5$:\hspace{0.1cm} &  & \hspace{0.1cm} [0] &  & $\begin{bmatrix} 0 & 0 \\  0 & 0 \end{bmatrix}$  &  \hspace{0.1cm} & [1]\hspace{0.1cm} \\
& & & & & &\\
\hline

\end{tabular}
\end{table*}

$$\hat{f}_{\rho_{(3)}} = [a_0],\quad \hat{f}_{\rho_{(2,1)}} = \begin{bmatrix} a_1 & a_2 \\  a_3 & a_4 \end{bmatrix},\quad \hat{f}_{\rho_{(1,1,1)}} = [a_5]$$

their inverse transform is, according to the previous proposition,
$$\mathcal{F}^{-1}(\hat{f}) = a_0 f_0 + a_1 f_1 + a_2 f_2 + a_3 f_3 + a_4 f_4 + a_5 f_5. $$

Therefore, any function $g:\Sigma_n\to\mathbb{R}$ is a linear combination of $f_0, f_1, \cdots, f_5$. Furthermore, the functions whose $(1,1,1)$ coefficient is null can be expressed as:
$$\mathcal{F}^{-1}(\hat{f}) = a_0 f_0 + a_1 f_1 + a_2 f_2 + a_3 f_3 + a_4 f_4. $$


Having the concept for $n = 3$, it is immediate to extend it to $n = 4$. In this case, the partitions of $n = 4$ are: $(4)$, $(3,1)$, $(2,2)$, $(2,1,1)$ y $(1,1,1,1)$. The dimensions of the irreducible representations are $d_{(4)}=1$, $d_{(3,1)}=3$, $d_{(2,2)}=2$, $d_{(2,1,1)}=3$ y $d_{(1,1,1,1)}=1$. Hence, proceeding as for $n = 3$, we can conclude that any function $f:\Sigma_4\to\mathbb{R}$ can be expressed as a linear combination of $d_{(4)}^2 + d_{(3,1)}^2 +d_{(2,2)}^2 + d_{(2,1,1)}^2 + d_{(1,1,1,1)}^2 = 1^2+3^2+2^2+3^2+1^2 = 24$ functions.\\

In other words, given $f:\Sigma_4\to\mathbb{R}$, there exist $a_0, a_1, \cdots, a_{23}\in\mathbb{R}$ such that
$$f = a_0f_0 + a_1f_1 + \cdots + a_{23} f_{23}$$

Suppose that, similarly to the case $n=3$, $f_{23}$ is the function for which ${(\hat f_{23})}_{\rho_{(1,1,1,1)}} = 1$ and the rest of the Fourier coefficients are 0. Then, for any other function $f$ with coefficient $\hat{f}_{\rho_{(1,1,1,1)}} = 0$, there exist $a_0, a_1, \cdots, a_{22}\in\mathbb{R}$ such that
$$f = a_0f_0 + a_1f_1 + \cdots + a_{22} f_{22}$$

Now, we may wonder what rankings can be generated by forcing the coefficient $(1,1,1,1)$ to be 0. Firstly, we should observe that the function $f_0$ ($f_0$ is the function which satisfies that all the Fourier coefficients are 0 except for $\hat{(f_0)}_{\rho_{(4)}} = 1$) does not add any information. This happens because 
$$f_0(\sigma) = \dfrac{1}{4!},\ \forall\sigma\in\Sigma_4$$
is constant. Therefore, we can discard this function and limit our study to the functions of the following form:
\begin{equation}\label{eq:comb_lineal}
f = a_1f_1 + a_2f_2 + \cdots + a_{22} f_{22}
\end{equation}

However, wondering which rankings will be generated is too general. So a first step consists in selecting a specific ranking $\tau\in\Sigma_{n!}$ and checking if $\tau$ can be generated with the functions that have a null $(1,1,1,1)$ coefficient. In other words, we would like to see if there exists a function $f$ that has the form given by equation (\ref{eq:comb_lineal}) and that also satisfies

$$
\left\{\begin{matrix}
f(\sigma_{\tau(1)}) > f(\sigma_{\tau(2)})\\ 
f(\sigma_{\tau(2)}) > f(\sigma_{\tau(3)})\\
\vdots\\ 
f(\sigma_{\tau(n!-1)}) > f(\sigma_{\tau(n!)})
\\ 
\end{matrix}\right.
$$ 

Using equation (\ref{eq:comb_lineal}), this system can be rewritten as
$$
\left\{\begin{matrix}
a_1f_1(\sigma_{\tau(1)}) + \cdots + a_{22}f_{22}(\sigma_{\tau(1)}) > a_1f_1(\sigma_{\tau(2)}) + \cdots + a_{22}f_{22}(\sigma_{\tau(2)})\\ 
a_1f_1(\sigma_{\tau(2)}) + \cdots + a_{22}f_{22}(\sigma_{\tau(2)}) > a_1f_1(\sigma_{\tau(3)}) + \cdots + a_{22}f_{22}(\sigma_{\tau(3)})\\
\vdots\\ 
a_1f_1(\sigma_{\tau(n!-1)}) + \cdots + a_{22}f_{22}(\sigma_{\tau(n!-1)}) > a_1f_1(\sigma_{\tau(n!)}) + \cdots + a_{22}f_{22}(\sigma_{\tau(n!)})
\\ 
\end{matrix}\right.
$$ 

Equivalently,

$$
\left\{\begin{matrix}
a_1(f_1(\sigma_{\tau(1)}) - f_1(\sigma_{\tau(2)})) + \cdots + a_{22}(f_{22}(\sigma_{\tau(1)}) - f_{22}(\sigma_{\tau(2)})) > 0\\ 
a_1(f_1(\sigma_{\tau(2)}) - f_1(\sigma_{\tau(3)})) + \cdots + a_{22}(f_{22}(\sigma_{\tau(2)}) - f_{22}(\sigma_{\tau(3)})) > 0\\
\vdots\\ 
a_1(f_1(\sigma_{\tau(n!-1)}) - f_1(\sigma_{\tau(n!)})) + \cdots + a_{22}(f_{22}(\sigma_{\tau(n!-1)}) - f_{22}(\sigma_{\tau(n!)})) > 0
\\ 
\end{matrix}\right.
$$ \\

One can also express it with the matricial form:
$$ \underbrace{\begin{bmatrix} 
f_1(\sigma_{\tau(1)}) - f_1(\sigma_{\tau(2)}) & f_2(\sigma_{\tau(1)}) - f_2(\sigma_{\tau(2)}) & \cdots & f_{22}(\sigma_{\tau(1)}) - f_{22}(\sigma_{\tau(2)}) \\  
f_1(\sigma_{\tau(2)}) - f_1(\sigma_{\tau(3)}) & f_2(\sigma_{\tau(2)}) - f_2(\sigma_{\tau(3)}) & \cdots & f_{22}(\sigma_{\tau(2)}) - f_{22}(\sigma_{\tau(3)}) \\
\vdots & \vdots & \vdots & \vdots \\
 f_1(\sigma_{\tau(23)}) - f_1(\sigma_{\tau(24)}) & f_2(\sigma_{\tau(23)}) - f_2(\sigma_{\tau(24)}) & \cdots & f_{22}(\sigma_{\tau(23)}) - f_{22}(\sigma_{\tau(24)})\\ 
\end{bmatrix}}_\text{$\mathrm{F}$}
\underbrace{\begin{bmatrix} a_1  \\  a_2 \\ \vdots \\ a_{22} \end{bmatrix}}_\text{$\mathbf{a}$} > 
\begin{bmatrix} 0  \\  0 \\ \vdots \\ 0 \end{bmatrix}
$$

To sum up, the problem can be reduced to finding whether the system $\mathrm{F}\mathbf{a} > 0$ has a solution.\\

For this purpose, we can refer to Gordan's alternative theorem (see \cite{mangasarian1981stable}):\\

\begin{theorem}[Gordan's alternative theorem]

Given a matrix $D\in\mathbb{R}^{m\times k}$, the following statements are equivalent:\\
\begin{enumerate}
\item There exists $\mathbf{y}\in\mathbb{R}^k$ such that $D\mathbf{y} > 0$.
\item $D^T\mathbf{v}= 0$, $0\ne\mathbf{v}\ge 0$ has no solution $\mathbf{v}\in\mathbb{R}^m$.
\end{enumerate}
\end{theorem}

If we apply this theorem to our present situation, we obtain:\\

\begin{corollary}
If $n = 4$, given a ranking $\tau\in\Sigma_{n!}$, there exists a function $f$ with a null $(1,1,1,1)$ coefficient that generates $\tau$ if and only if the system $\mathrm{F}^T\mathbf{v}= 0$ has no solution $0\ne\mathbf{v}\ge 0$.
\end{corollary}

\subsection{What happens when one tries to generalize the method?} 

In the end, when applying this method, we have to solve the system $\mathrm{F}\in\mathbb{R}^{m\times k}$, where $m$ is $(n!)!-1$ and $k$ is the number of ``base functions" that we choose (that is, the sum of the squares of $d_{\rho_\lambda}$, for all the non-null coefficients $\lambda$, except for the trivial coefficient $(n)$). \\

If we are using the same non-null coefficients as in the QAP, then we have only 4, if $n = 3$, and the dimensions of the system are $m\times k = 5\times 4$. If $n = 4$, the dimesions of the final system are $m\times k = 23 \times 22$. In both cases, we are lucky, because, even though the equation $\mathrm{F}^T\mathbf{v} = 0$ has an infinite number of solutions, the space of solutions has dimension $m-k = 1$, which means that it is generated by a single vector and we only have to check that.\\

But the application of the method is difficult for higher dimensions. Even for $n = 5$ it gets complicated, since $m\times k = 119\times 77$. The space of solutions of $\mathrm{F}^T\mathbf{v} = 0$ has dimension $m-k$, which becomes hardly controllable. Among these solutions, we can discard the ones involving complex numbers, but even so, in all the specific cases that we have tried, it is still too complicated to have a systematic approach. In addition, as $n$ grows, the difference between the dimensions of the system grows further, because we are only considering 4 coefficients and $(n!)!-1$ grows really rapidly. \\


\begin{acknowledgements}
This work has been partially supported by Spanish Ministry of Economy, Industry and Competitiveness (TIN2016-78365R). Jose A. Lozano is also supported by Basque Government through BERC 2018-2021 and Elkartek programs, and by Spanish Ministry of Economy and Competitiveness MINECO: BCAM Severo Ochoa excellence accreditation SEV-2017-0718. Anne Elorza holds a predoctoral grant (ref. PIF17/293) from the University of the Basque Country.
\end{acknowledgements}

\newpage
\bibliographystyle{spmpsci}      
\bibliography{bib_firstFourierPaper}
%

\end{document}